\documentclass[10pt,journal,twocolumn,final]{IEEEtran}
\usepackage{amssymb,amsmath,amsthm}
\usepackage[utf8]{inputenc}
\usepackage[T1]{fontenc}
\usepackage{graphicx}
\usepackage{epsfig}
\usepackage{amsmath}
\usepackage[english]{babel}
\usepackage{subfigure}
\usepackage{color}
\usepackage{epstopdf}
\usepackage{import}
\newtheorem{lemma}{Lemma}

\makeatletter
\newcommand*{\rom}[1]{\expandafter\@slowromancap\romannumeral #1@}
\makeatother

\begin{document}

\title{Simultaneous Wireless Information and Power Transfer Under Different CSI Acquisition Schemes}

\author{
\authorblockN{Chen-Feng Liu, Marco Maso,~\IEEEmembership{Member,~IEEE,} Subhash Lakshminarayana,~\IEEEmembership{Member,~IEEE,} \\Chia-Han Lee,~\IEEEmembership{Member,~IEEE,} and Tony Q. S. Quek,~\IEEEmembership{Senior Member,~IEEE}}

\thanks{This work was supported in part by the SRG ISTD 2012037, SUTD-MIT International Design Centre under Grant IDSF1200106OH, and the SUTD-ZJU Research Collaboration Grant, and the Ministry of Science and Technology (MOST), Taiwan under Grant MOST103-2221-E-001-002.}

\thanks{C.-F. Liu and C.-H. Lee are with the Research Center for Information Technology Innovation, Academia Sinica, Taipei 115, Taiwan (email: cfliu@citi.sinica.edu.tw; chiahan@citi.sinica.edu.tw).}

\thanks{M. Maso is with the Mathematical and Algorithmic Sciences Lab, Huawei France Research Center, Boulogne-Billancourt 92100, France (email: marco.maso@huawei.com).}

\thanks{S. Lakshminarayana is with the Singapore University of Technology and Design, Singapore 138682 (email: subhash@sutd.edu.sg).}

\thanks{T. Q. S. Quek is with the Singapore University of Technology and Design, Singapore 138682, and also with the Institute for Infocomm Research, A*STAR, Singapore 117685 (email: tonyquek@sutd.edu.sg).}}

\maketitle

\begin{abstract}
In this work, we consider a multiple-input single-output system in which an access point (AP) performs a simultaneous wireless information and power transfer (SWIPT) to serve a user terminal (UT) that is not equipped with external power supply. In order to assess the efficacy of the SWIPT, we target a practically relevant scenario characterized by imperfect channel state information (CSI) at the transmitter, the presence of penalties associated to the CSI acquisition procedures, and non-zero power consumption for the operations performed by the UT, such as CSI estimation, uplink signaling and data decoding. We analyze three different cases for the CSI knowledge at the AP: no CSI, and imperfect CSI in case of time-division duplexing and frequency-division duplexing communications. Closed-form representations of the ergodic downlink rate and both the energy shortage and data outage probability are derived for the three cases. Additionally, analytic expressions for the ergodically optimal duration of power transfer and channel estimation/feedback phases are provided. Our numerical findings verify the correctness of our derivations, and also show the importance and benefits of CSI knowledge at the AP in SWIPT systems, albeit imperfect and acquired at the expense of the time available for the information transfer.
\end{abstract}


\section{Introduction} \label{sec:introduction}

\IEEEPARstart{I}{n} conventional wireless systems, the limited battery capacity of mobile devices typically affects the overall network lifetime. Increasing the size of the battery might not be a feasible solution to address this problem, due to a consequent reduction of the portability and increase in the cost of the equipment. For these reasons, the study of novel techniques to prolong the lifetime of the battery has triggered an increased interest in the wireless communications community. In this context, the study of the so-called wireless power transfer (WPT) has recently gained prominence as means to implement a cable-less power transfer between devices \cite{WPTOverview}, either by resonant inductive coupling \cite{Kurs06072007} or by far-field power transfer \cite{art:brown84}. The latter approach has seen increasing momentum in recent years, due to its promising potential for longer range transfers. A fundamental breakthrough in this context has been the design of rectifying antennas (rectennas) for microwave power transfer (MPT), key component to achieve an efficient radio frequency to direct current (RF-to-DC) conversion. This has brought to several technological advances, e.g., the design of flying vehicles powered solely by microwave \cite{conf:Schlesak88}, which confirmed that RF-to-DC conversion can not only be performed but also achieve remarkable efficiency. In this regards, commercial products exhibiting an efficiency larger than $50\%$ are already available on the market \cite{prod:P2110}. Remarkably, performance of state-of-the-art RF-to-DC converters can be even higher, i.e., more larger than $80\%$, resulting in an impressive potential DC-to-DC efficiency of $45\%$ \cite{art:le2008, art:mcspadden02}.

\subsection{Related Works}

Recent advances in signal processing and microwave technology have shown that far-field power transfer can offer interesting perspectives also in the context of traditional wireless communication systems. For instance, the same electromagnetic field could be used as a carrier for both energy and information, realizing the so-called simultaneous wireless information and power transfer (SWIPT). The potential of this approach has been first highlighted by the studies proposed in \cite{Varshney} and \cite{ShannonTeslar}. Therein the trade-off between the two transfers within SWIPT was investigated, in the case of both flat fading and frequency-selective channels. After these pioneering works, several studies have been performed to assess the implementability of receivers that can make use of RF signals to both harvest energy and decode information \cite{art:hui13}, and analyze the performance of SWIPT in many scenarios, e.g., multi-antenna systems \cite{art:zhang13}, opportunistic networks \cite{art:kaibin13}, wireless sensor networks \cite{art:visser13}. In particular, the aforementioned trade-off is investigated in multiple-input multiple-output (MIMO) systems for two information and power transfer architectures, i.e., time-switching and power-splitting, for both perfect \cite{art:zhang13} and imperfect \cite{imperfectCSI} channel state information (CSI). A different approach is considered in \cite{ReviewerRef2}, where the trade-off between the information and the energy transfer is investigated in a multi-user network under two different constraints, i.e., a constraint expressed in terms of secrecy rate for the former and amount of harvested energy at the receiver for the latter.

A line of work considering orthogonal frequency division multiplexing (OFDM) systems is presented in \cite{NgOFDMSingleSplitting, NgOFDMMultiSeparate, NgOFDMMultiSplitting}, where the resource allocation policy at the transmitter is studied for both single and multi-user scenarios, in which the receiver adopts a power-splitting strategy and considers undesired interference as an additional energy resource. Finally, ad-hoc scenarios departing from the standard SWIPT paradigm have been analyzed in works such as \cite{ChenEnergyBeamforming, ReviewerRef1}, where it is assumed that the energy and information transfers are performed by two different devices, operating in frequency-division duplexing (FDD) mode. More precisely, in \cite{ChenEnergyBeamforming} the time allocation policy for the two transmitters is studied, under the assumption that the efficiency of the energy transfer is maximized by means of an energy beamformer that exploits a quantized version of the CSI received in the uplink by the energy transmitter. Along similar lines, \cite{ReviewerRef1} investigates the optimal time and power allocations strategies such that the amount of harvested energy is maximized, taking into account the impact of the CSI accuracy on the latter quantity.

\subsection{Summary of Our Contribution}

In general, most existing works for SWIPT rely on ideal assumptions: a) availability of perfect CSI at the transmitter, b) no penalty for CSI acquisition, c) no power consumption for signal decoding operations at the receiver. If these assumptions are relaxed and the resulting penalties and issues are considered, then the performance of wireless systems can decrease significantly. Studies taking into account these aspects have been proposed for conventional information transfers, and the effect of imperfect CSI acquisition, training, feedback as well as the resource allocation problem have been thoroughly analyzed \cite{CaireFeedback}. Departing from these observations, in this work we aim at investigating the efficacy of SWIPT for practically relevant scenarios, by relaxing the three aforementioned ideal assumptions. In particular, we target a multiple-input single-output (MISO) system consisting of a multi-antenna access point (AP) which transfers both information symbols and energy to a single user terminal (UT) that does not have access to any external power source. Accordingly, in contrast to the previous contributions on this topic, in this work the power harvested by means of the WPT is used by the UT to perform all the necessary signal processing operations for both information decoding and uplink communications. Additionally, we adopt a systematic approach and consider the three main possible scenarios for an AP that engages in a downlink transmission in modern networks, to provide a more complete characterization of the considered system. Consequently, in this work, the performance and feasibility of the SWIPT in a MISO system is studied for the three following cases:
\begin{itemize}
\item No CSI available at the AP.
\item Time-division duplexing (TDD) communications and CSI acquisition at the AP by means of training symbols.
\item FDD communications and CSI acquisition at the AP by means of analog symbols feedback.
\end{itemize}
We compare these three scenarios for three performance metrics of interest, namely, ergodic downlink rate, energy shortage probability, and data outage probability. Our contributions in this work are as follows:
\begin{itemize}
\item We derive closed-form representations for the three performance metrics of interest in all three scenarios and match them to the numerical results.
\item We derive the approximations of the ergodically optimal duration of the WPT phase in all the three scenarios as a portion of the channel coherence time.
\item Additionally, for the TDD and the FDD scheme, we derive closed-form approximations for the ergodically optimal duration of the channel training/feedback phases, to maximize the downlink rate.
\item We show that the TDD scheme can outperform the FDD scheme in SWIPT systems in terms of both downlink rate and data outage probability.
\end{itemize}
Our numerical findings verify the correctness of our derivations. More specifically, concerning the downlink rate, we show that the performance gap between the numerical optimal solutions and the results obtained by means of our approximations is very small for low to mid signal-to-noise ratio (SNR) values and negligible for high SNR values. Moreover, we show that both TDD and FDD outperform the non-CSI case at any SNR value in terms of downlink rate. This confirms that CSI knowledge at the AP is always beneficial for the information transfer in SWIPT systems, despite both its imperfectness and the resources devoted to the channel estimation/feedback procedures. The correctness of our derivations is further verified when numerically evaluating both the energy shortage and data outage probability of the considered MISO system adopting SWIPT, for which a perfect match of analytic and numerical results is achieved. Finally, it is worth noting that throughout our study TDD consistently outperforms FDD in terms of both downlink rate and data outage probability, confirming the potential of this duplexing scheme for the future advancements in modern networks. 

The rest of the paper is organized as follows. In Sec.~\ref{Sec: System Setup}, we specify the system model. In Sec.~\ref{Sec: Non-CSI scheme}, \ref{Sec: TDD scheme}, and \ref{Sec: FDD scheme}, we specify the system model and derive the downlink rate for the non-CSI, the TDD, and the FDD scheme. In Sec.~\ref{Sec: Outage probability}, we derive the energy shortage and data outage probability for each scheme. In Sec.~\ref{Sec: Numerical results}, we show and discuss the numerical results. Finally, we conclude in Sec.~\ref{Sec: Conclusions}. 

{\it Notations:} 
In this paper, we denote matrices as boldface upper-case letters, vectors as boldface lower-case letters. Additionally, we let $[\cdot]^\dag$ be the conjugate transpose of a vector. All vectors are columns, unless otherwise stated. Furthermore, for a scalar $c \in \mathbb{C}$ we note by $c^*$ its complex conjugate. We use $\mathbf{x} \perp \mathbf{y}$ to express the orthogonality between vectors $\mathbf{x}$ and $\mathbf{y}$. We denote a circular symmetric Gaussian random vector with mean $\boldsymbol{\mu}$ and covariance matrix $\mathbf{\Sigma}$ as $\mathcal{C}\mathcal{N}\left(\boldsymbol{\mu},\mathbf{\Sigma}\right)$. The chi-squared distribution with $K$ degrees of freedom is denoted by $\chi_K^2$ and its probability density function (PDF)  is given by $f_{X}\left(x\right)=\frac{1}{\Gamma\left(\frac{K}{2}\right)}2^{-K/2}x^{\frac{K}{2}-1}e^{-x/2}$ where $\Gamma(q)=\int_{0}^{\infty}u^{q-1}e^{-u}\mathrm{d}u$ is the Gamma function. The non-central chi-squared distribution with $K$ degrees of freedom and non-central parameter $\nu$ is denoted by $\chi_K^{'2}\left(\nu\right)$ and its PDF is given by
$f_{X}\left(x \right)=\frac{1}{2}e^{-\left(x+\nu\right)/2}\left(\frac{x}{\nu}\right)^{K/4-1/2}I_{\frac{K}{2}-1}\left(\sqrt{\nu x}\right)$, 
$I_n\left(\cdot\right)$, modified Bessel function of the first kind. 
$\Gamma(q,r)=\int_{r}^{\infty}u^{q-1}e^{-u}\mathrm{d}u$ is the upper incomplete Gamma function.
$Q_M\left(q,r\right)=\int_{r}^{\infty}\frac{\xi^M}{q^{M-1}}\exp\left(-\frac{\xi^2+q^2}{2}\right)I_{M-1}\left(q\xi\right)\mathrm{d}\xi$ is the generalized Marcum Q-function \cite{art:marcum50}.

\section{System Model}\label{Sec: System Setup}

We consider a point-to-point communication system consisting of an AP with $L$ antennas and a UT with single antenna.  We denote the downlink channel (from the AP to the UT) as $\mathbf{h}=[h_1,\cdots,h_L]^\top$. The channel is assumed to be block fading, with independent fading from block to block. The entries of the channel vector are complex Gaussian (Rayleigh fading), hence $\mathbf{h} \sim \mathcal{CN}\left(\mathbf{0},\mathbf{I}_L\right)$. Let $T_C$ be the coherence time length. For simplicity in the notation, we assume that the total number of symbols that can be transmitted within the coherence time is $T_C$. The AP transmits the symbol $\mathbf{x} \in \mathbb{C}^{L \times 1}$ with a transmit power $P$, i.e., $\mathbb{E}\left[\lVert\mathbf{x}\rVert^2\right]=P$. The received signal at the UT is given by
$y = \mathbf{h}^{\dag}\mathbf{x}+n,$
where $n\sim \mathcal{CN}\left(0,N_0\right)$ is the thermal noise, modeled as a complex additive white Gaussian noise (AWGN). We assume the UT does not have any external power source (such as the battery) and all the power required for the operations to be performed at the UT is provided by the AP through the WPT component of the SWIPT. Accordingly, the UT is equipped with a circuit that can perform two different functions: a) harvest energy from the received RF signal, b) information decoding. As considered in previous literature \cite{art:zhang13}, we assume that the UT cannot harvest energy and decode information from the same signal, at the same time. Hence, a time switching strategy is adopted under which the AP transmits the signals in two phases: the signal sent during the first phase has WPT purposes and is used by the UT to harvest energy, whereas the signal sent in the second phase has information transfer purposes. Note that, throughout this work, we assume that the energy harvested in the first phase (power transfer phase) is the sole source of power for all the subsequent operations performed by the UT (the exact details of these operations will be specified the following sections).\footnote{Note that, the words energy and power are used interchangeably in this paper, for the sake of simplicity, in spite of their conceptual difference.}

\subsubsection{Details of the Power Transfer Phase}

During each coherence time interval, the AP first transmits the power wirelessly to the UT for $\epsilon< T_C$\footnote{The exact value of $\epsilon$ will be specified later depending upon the mode of operation.} time slots. First, the AP divides its power $P$ equally between its $L$ transmit antennas to perform the WPT. Hence the $L$-sized transmit symbol during this phase, denoted by ${\bf x}_{\text{EH}}$, is given by
${\bf x}_{\text{EH}} = \sqrt{\frac{P}{L}} {\bf s},$
where ${\bf s}$ is a random vector with zero mean and covariance matrix $\mathbb{E}\left[{\bf s}{\bf s^{\dag}}\right]=\mathbf{I}_L$. Thus, the power harvested at the UT is given by
\begin{equation} \label{Eq: System-Harvested power}
P_H=\frac{\beta P\lVert\mathbf{h}\rVert^2}{L},
\end{equation}
where $\beta \in [0,1]$ is a coefficient that measures the efficiency of the RF to direct current (RF-to-DC) power conversion \cite{art:zhang13, art:visser13}.

\subsubsection{Details of the Information Transmission Phase}

For the second phase, namely information transmission phase, we adopt a systematic approach and consider the three scenarios. In the first one, the AP transmits the information symbols without the knowledge of CSI (we will refer to this approach as non-CSI scheme). In the second scheme, we consider a TDD communication, in which the downlink and uplink communications are performed over the same bandwidth. Accordingly, first the AP acquires the CSI by evaluating a pilot sequence transmitted by the UT in the uplink, and then engages in the downlink transmission. In the last case, we consider an FDD communication, in which the downlink and uplink communications are performed over two separate bandwidths. Consequently, in this case, UT sends an analog feedback signal in the uplink, carrying the downlink channel estimation, to allow the AP to acquire the CSI and subsequent transmit information symbols.

Under the aforementioned settings, we analyze the performance of the system for two different metrics of interest, namely the downlink rate and outage probability. We provide a detailed analysis of these two metrics in the rest of the paper.

\section{Analysis of the Downlink Rate}\label{sec:downlink}

In this section, we analyze the downlink rate for the three considered schemes.

\subsection{Non-CSI Scheme}\label{Sec: Non-CSI scheme}

We first consider the case where the AP transmits the information symbols without the knowledge of CSI. The schematic diagram of this scenario is shown in Fig.~\ref{Fig: Non CSI - System model}.
\begin{figure}[!h]
\centering
\subfigure[Operations of the AP.]
{\def\svgwidth{.67\columnwidth}
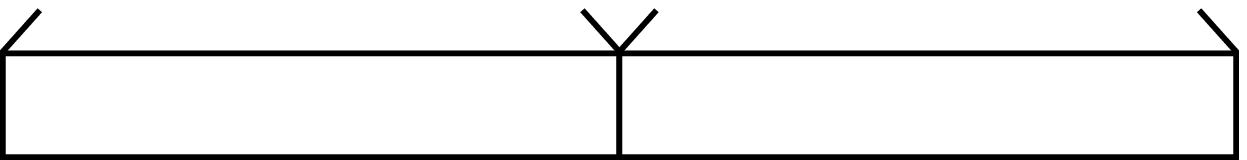}
\label{Fig: Non CSI - Operation of AP} 
\subfigure[Operations of the UT.]
{\def\svgwidth{.67\columnwidth}
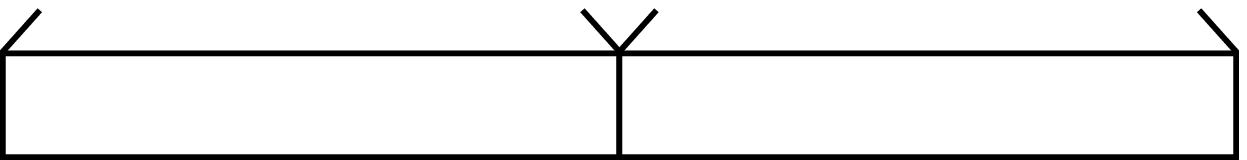}
\label{Fig: Non CSI - Operation of UT} 
\caption{Operations of the AP and the UT during the coherence time in the non-CSI scheme.}
\label{Fig: Non CSI - System model}
\end{figure}
Under this scheme, the system utilizes $\epsilon = \alpha_NT_C$ symbols to transfer power
and the remaining symbols to transmit information symbols, where $0<\alpha_N< 1$. 
The received signal during the information transmission phase
is given by
$y = \mathbf{h}^\dagger \mathbf{x} + n,$
where the $\mathbf{x} = \sqrt{\frac{P}{L}} \mathbf{s}$ and 
$\mathbf{s} \sim \mathcal{CN}(\mathbf{0},\mathbf{I}_L).$
Note that in the absence of CSI, the AP performs equal power allocation over all its antennas to transmit the information
symbol.
For the information decoding at the UT, we consider that the power consumption of the circuit components devoted to the decoding is proportional to the number of received symbols (as typically considered in previous works on the subject \cite{conf:heinzelman00}). Accordingly, we denote the power consumption per decoded symbol at the UT as $P_D$.

Since the power harvested in the first phase must be sufficient to decode all the information symbols, we have that
$\alpha_NT_C P_H=\left(1-\alpha_N\right)T_C P_D.$
Now, if we plug \eqref{Eq: System-Harvested power} into this equation then, after some manipulations, we have that the minimum fraction of time that should be devoted to the power transfer, i.e., $\alpha_N$, given by
\begin{equation}\label{Eq: Non-Time portion alpha}
\alpha_N=\frac{LP_D}{\beta P\lVert\mathbf{h}\rVert^2+LP_D}.
\end{equation}
We now analyze the downlink rate for this scheme. We recall that the AP can transmit $1-\alpha_N$ symbols for the information transfer. Accordingly, using \eqref{Eq: Non-Time portion alpha}, the downlink rate obtained for the non-CSI scheme is given by
\begin{align}
R_{NC}& =R_{NC}(\alpha_N)=\left(1-\alpha_N\right)\log_2\left(1+\frac{P\lVert\mathbf{h}\rVert^2}{N_0L}\right) \notag
\\&=\frac{\beta P\lVert\mathbf{h}\rVert^2}{\beta P\lVert\mathbf{h}\rVert^2+LP_D}\log_2\left(1+\frac{P\lVert\mathbf{h}\rVert^2}{N_0L}\right).\label{Eq: Non-Rate}
\end{align}

\subsection{TDD Scheme}\label{Sec: TDD scheme}

We switch our focus to the TDD scheme, whose schematic diagram is shown in Fig.~\ref{Fig: TDD-System model}.
\begin{figure}[!h]
\centering
\subfigure[Operations of the AP.]{\def\svgwidth{\columnwidth}
 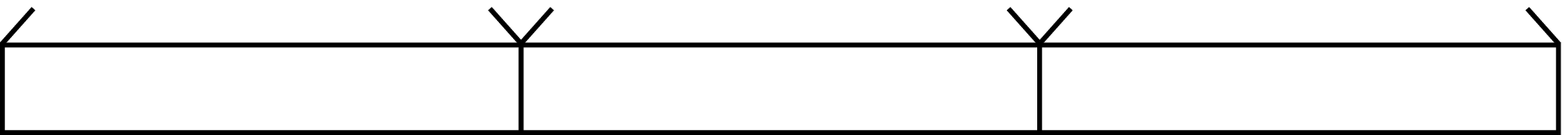}
 \label{Fig: TDD-Operation of AP} 
\subfigure[Operations of the UT.]{\def\svgwidth{\columnwidth}
 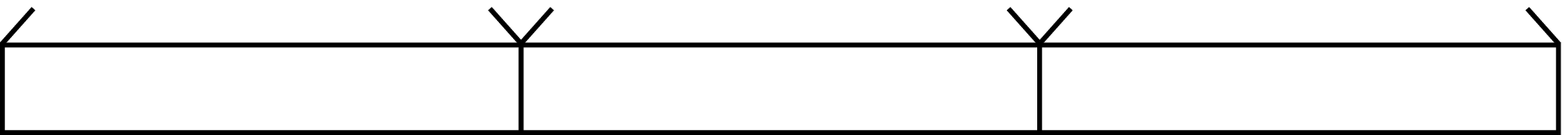}
 \label{Fig: TDD-Operation of UT} 
\caption{Operations of the AP and the UT during the coherence time in the TDD scheme.}
\label{Fig: TDD-System model}
\end{figure}
We recall that, in this case, the AP should provide the UT with sufficient energy for the latter to be able not only to decode the received data but also to perform all the operations related the uplink signaling inherent to the TDD scheme. Accordingly, the system utilizes $\epsilon = \alpha_TT_C$ symbols for power transfer, where $0<\alpha_T<1$. This is followed by the CSI acquisition phase. Since we assume that all the operations are performed within the coherence time, the AP can exploit the reciprocity of the downlink and uplink channels, inherent feature of the TDD scheme. This way, the channel estimated in the uplink can be used to design the beamformer for the downlink transmission. Accordingly, the UT transmits uplink pilots with power $P_E$ for the next $\eta_TT_C\in\mathbb{Z}^{+}$ symbol periods, with $0<\eta_T<1$ and $0<\alpha_T+\eta_T\leq 1$. The signal received by the AP during the $i$th symbol period (in the uplink pilot transmission phase) is given by
$\mathbf{y}^{p}_{T}[i]=\sqrt{P_E}\mathbf{h}^{*}+\mathbf{w}[i],$
where $\mathbf{w}[i] \sim \mathcal{C}\mathcal{N}(\mathbf{0},N_0\mathbf{I}_L)$ is the Gaussian noise at the AP. The AP estimates the channel by a minimum variance unbiased (MVU) based estimator \cite{KayEstimation}. Thus, the channel estimate at the AP is given by
\begin{align}
\mathbf{\hat{h}} & =\frac{1}{\sqrt{P_E}\eta_TT_C}\sum\limits_{i=1}^{\eta_TT_C}\left(\sqrt{P_E}\mathbf{h}+\mathbf{w}^*[i]\right)  =\mathbf{h}+\bar{\mathbf{w}},\label{Eq: TDD-Estimated channel at AP}
\end{align}
where $\mathbf{\bar{w}}\sim\mathcal{C}\mathcal{N}\left(\mathbf{0},\frac{N_0}{\eta_TT_CP_E}\mathbf{I}_L\right)$ denotes the estimation error.

This is followed by the information transmission phase. The focus of this work is on the performance of the SWIPT under non-ideal system assumptions and practical transmit schemes. Therefore, for simplicity we will assume that the AP beamforms the signal carrying the information symbols with a matched filter precoder (MFP) \cite{conf:demig2008}, optimal linear filter for maximizing the SNR. Accordingly, the AP exploits the CSI estimate to design the desired beamforming vector, obtained as
$\mathbf{m}_T =  \mathbf{\hat{h}}/\lVert\mathbf{\hat{h}}\rVert.$
Then, the received signal at the UT is given by
\begin{equation}
y_{T}^{s}=\frac{\mathbf{h}^{\dag}\mathbf{\hat{h}}}{\lVert\mathbf{\hat{h}}\rVert}s+n,\label{Eq: TDD-Received information}
\end{equation}
where $s$ is the information symbol, with $\mathbb{E}\left[\lvert s\rvert^2\right]=P$. As in the previous case, the UT consumes a power $P_D$ to decode every received information symbol. Thus, since the harvested power by the UT must be sufficient to send the pilot symbols and decode information at the UT, the condition
$\alpha_TT_C P_H=\eta_TT_C P_E+\left(1-\alpha_T-\eta_T\right)T_CP_D$
must be satisfied for the TDD scheme. Now, if we plug \eqref{Eq: System-Harvested power} into this condition then, after some manipulations, we have that the minimum fraction of time that should be devoted to the power transfer, i.e., $\alpha_T$, is given by
\begin{equation} \label{Eq: TDD-Time portion alpha}
\alpha_T=\frac{\eta_TLP_E-\eta_TLP_D+LP_D}{\beta P\lVert\mathbf{h}\rVert^2+LP_D}.
\end{equation}
Accordingly, we can use \eqref{Eq: TDD-Time portion alpha} to compute the downlink rate for the TDD scheme as
\begin{align}
&R_T =R_T (\alpha_T,\eta_T)=\left(1-\alpha_T-\eta_T\right)\log_2\left(1+\frac{P\lvert\mathbf{h}^{\dag}\mathbf{\hat{h}}\rvert^2}{N_0\lVert\mathbf{\hat{h}}\rVert^2}\right)\notag
\\&=\frac{\left(1-\eta_T\right)\beta P\lVert\mathbf{h}\rVert^2-\eta_TLP_E}{\beta P\lVert\mathbf{h}\rVert^2+LP_D}\log_2\left(1+\frac{P\lvert\mathbf{h}^{\dag}\mathbf{\hat{h}}\rvert^2}{N_0\lVert\mathbf{\hat{h}}\rVert^2}\right). \label{Eq: TDD-Rate}
\end{align}
Note that the channel training time $\eta_T T_c$ impacts both the accuracy of the estimated channel vector and the remaining available time for information transfer. As a consequence, let $\eta^{\star}_T$ be the duration of the portion of coherence time devoted to the channel training/estimation that maximizes the ergodic downlink rate, defined as
\begin{align}
\eta^{\star}_T = \operatorname{argmax}_{\eta_T}
\mathbb{E}_{\mathbf{h},\mathbf{\bar{\mathbf{w}}}}\Bigg[\frac{\left(1-\eta_T\right)\beta P\lVert\mathbf{h}\rVert^2-\eta_TLP_E}{\beta P\lVert\mathbf{h}\rVert^2+LP_D}&\notag
\\\times\log_2\left(1+\frac{P\lvert\mathbf{h}^{\dag}\mathbf{\hat{h}}\rvert^2}{N_0\lVert\mathbf{\hat{h}}\rVert^2}\right)&\Bigg].\label{Eq: SecIII-TDD-rate optimal problem}
\end{align}
The derivation of the exact value of $\eta^{\star}_T$ is very complicated. However, two approximations of this value, valid for high and low SNR, respectively, can be derived as stated in the following result.

\begin{lemma}\label{Lem: TDD eta over all channels}

At high SNR, $\eta^{\star}_T$ can be approximated as
\begin{equation}
\eta_T^{\star} \approx\sqrt{\frac{N_0L\log_2e}{B_1T_CP_E\left(L-1\right)}},\label{Eq: TDD_opt_high_SNR}
\end{equation}
where
$
B_1=\mathbb{E}_{\mathbf{h}}\left[\left(1+\frac{LP_E}{\beta P\lVert\mathbf{h}\rVert^2}
\right)\log_2\left(\frac{P\lVert\mathbf{h}\rVert^2}
{N_0}\right)\right].
$
At low SNR, it can be approximated as
\begin{align}
\eta_T^{\star} &\approx\frac{N_0}{T_CP_E}
\left(-1+\sqrt{1+\frac{\left(L-1\right)\beta PT_CP_E}{LN_0\left(\beta P+P_E\right)}-\frac{1}{L}}\right).\label{Eq: TDD_opt_low_SNR}
\end{align}
\end{lemma}
\begin{proof}
See Appendix-\ref{Apx: TDD eta over all channels}.
\end{proof}
Lemma \ref{Lem: TDD eta over all channels} provides a result whose interpretation is not trivial. In fact, several parameters are present in \eqref{Eq: TDD_opt_high_SNR} and \eqref{Eq: TDD_opt_low_SNR}, thus understanding their impact on the accuracy of the proposed approximations is rather complex. However, some interesting insights can drawn from Lemma \ref{Lem: TDD eta over all channels}, if we focus on the approximations that are introduced in order to derive the final results. First, we note that the impact of $P_D$ on the accuracy of the results is likely negligible, due the fact that $P \gg P_D$ by construction. Subsequently, let us focus on the quantity $\lambda=\frac{2\eta_TT_CP_E\lVert\mathbf{h}\rVert^2}{N_0}$, introduced in \eqref{Eq: TDD eta high SNR Taylor series-1}. If we fix $N_0$ at the denominator of $\lambda$ then it is straightforward to see that the latter increases with an increase in $P_E$ and $L$. Now, consider the low SNR case. In this case, $N_0$ is very large, thus the approximation $\lambda \approx 0$ is adopted. In practice, the accuracy of this approximation depends on the value of $L$ and $P_E$, i.e., the lower those values are, the more accurate the approximation is. Switching our focus to the high SNR analysis, we observe an opposite behavior. In fact, in this case $N_0$ is very small, hence the approximation $\lambda \gg 0$ is introduced. Thus, the accuracy of the approximation is greater when $P_E$ and $L$ are large. Concerning the latter parameter, i.e., $L$ number of antennas at the AP, we note that an analysis of its impact on the accuracy of the results in Lemma \ref{Lem: TDD eta over all channels} is extremely interesting, given its relevance in a MISO systems. Accordingly, a detailed discussion on this aspect will be provided in Sec.~\ref{Sec: Numerical results}.

\subsection{FDD Scheme}\label{Sec: FDD scheme}

We now consider the FDD scheme, whose schematic diagram is illustrated in Fig.~\ref{Fig: FDD-System model}.
\begin{figure}[!h]
\centering
\subfigure[Operations of the AP.]{\def\svgwidth{.67\columnwidth}
 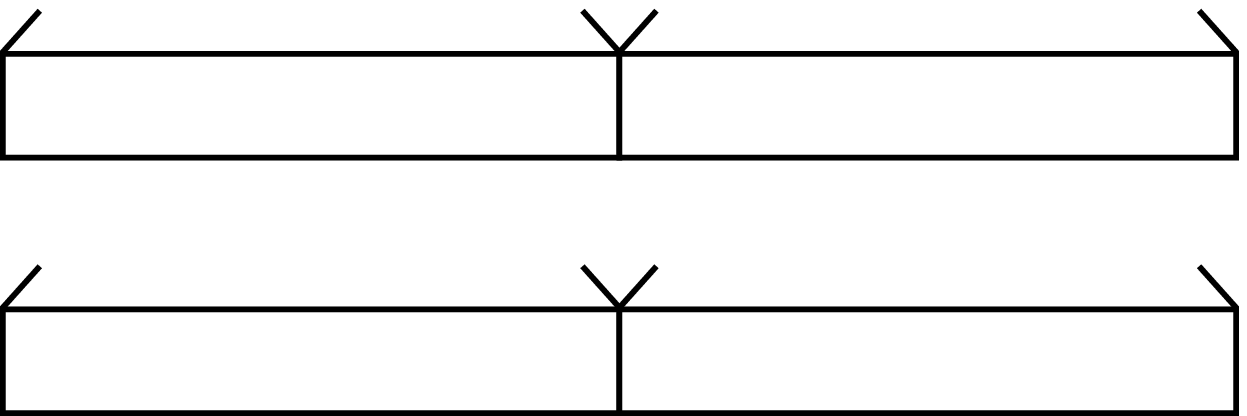}
\vspace{1mm}
\subfigure[Operations of the UT.]{\def\svgwidth{.67\columnwidth}
 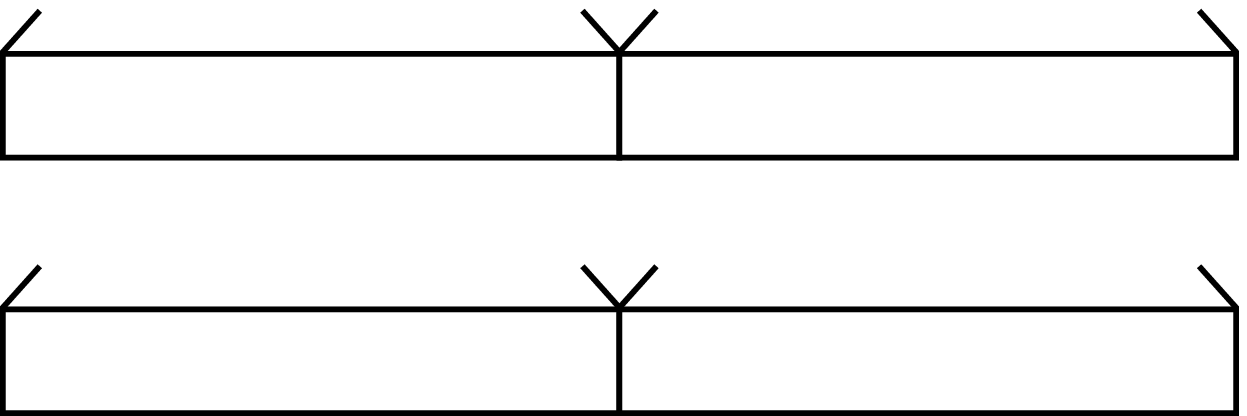}
 \caption{Operations of the AP and the UT during the coherence time in the FDD scheme.}
 \label{Fig: FDD-System model}
\end{figure}
%
%
%\noindent
Differently from the TDD case, the downlink and uplink channels are in general uncorrelated in the FDD scheme. Therefore, two separate channel estimation procedures have to be performed at the UT and the AP, to be able to provide to the latter with the CSI w.r.t. the downlink channel. Accordingly, the operations in the FDD scheme are as follows. First, the AP transfers power to the UT for a period of $\epsilon = \alpha_FT_C$, with $0<\alpha_F<1$. As before, we recall that the AP should provide the UT with sufficient energy for the latter to be able not only to decode the received data but also to perform all the operations related the uplink signaling inherent to the FDD scheme. Afterwards, a downlink channel training phase takes place, in which the AP sends pilot sequences of $\eta_FT_C\in\mathbb{Z}^{+}$ symbols with power $P$ to the UT for estimating the downlink channel, with $0<\eta_F<1$. Finally, the UT feeds back in the uplink the estimated CSI in analog form over the subsequent $\tau_FT_C\in\mathbb{Z}^{+}$ symbols, where $0<\tau_F<1$ and $0<\alpha_F+\eta_F+\tau_F\leq 1$. Note that, in this work we adopt a simplified model for the uplink communication, for the sake of simplicity of the analysis, and matters of space economy. Specifically, we assume that the feedback signal sent by the UT to the AP experiences an AWGN channel. We note that, this follows the typical approach proposed in the literature for first studies on CSI acquisition schemes based on analog feedback signals \cite{CaireFeedback, art:samardzija06}.

Now, let us analyze the aforementioned steps in detail. Consider the $l$th antenna. We denote the pilot sequence sent over it as $\mathbf{e}_l=[e_l[1] \cdots,e_l[\eta_FT_C]]^\top\in\mathbb{C}^{\eta_F T_C}$, $l \in [1,L]$. Naturally, the sequences adopted in this phase are known at both ends of the communications. In particular, without loss of generality, we assume orthogonality between pilot sequences sent over different antennas, i.e., $\mathbf{e}_i \perp \mathbf{e}_j,$ for $i \neq j$. Thus, in order to guarantee their orthogonality, and estimate $L$ independent channel coefficients, a lower bound on the minimum sequence size must be satisfied, i.e., $\eta_FT_C\geq L$. Moreover, the AP equally divides the power $P$ among its $L$ antennas, yielding $||e_l|| ^2 = \frac{P}{L}$ and thus $\lVert\mathbf{e}_{1}\rVert^2=\cdots=\lVert\mathbf{e}_{L}\rVert^2=\frac{\eta_FT_CP}{L}$. Then, the signal received by the UT during the downlink channel training phase is given by
$\mathbf{y}_{UT,F}^{p} = \mathbf{e}_{1}h_{1}^{*}+\cdots+\mathbf{e}_{L}h_{L}^{*}+\mathbf{w}_{UT},$
where $\mathbf{w}_{UT}\sim\mathcal{C}\mathcal{N}\left(\mathbf{0},N_0\mathbf{I}_{\eta_FT_C}\right)$ is the thermal noise at the UT.
The UT in turn multiplies the received signal $\mathbf{y}_{UT,F}^{p}$ by $\mathbf{e}_{l}^{\dag}/\lVert\mathbf{e}_{l}\rVert^2$ to estimate the $l$th channel coefficient, $h_l$. Similar to the previous section, the $L$ downlink channel coefficients are estimated by an MVU based estimator. The estimated channel vector at the UT can be written as
$\mathbf{\hat{h}}_{UT}=[\hat{h}_{UT,1}, \dots, \hat{h}_{UT,L}]^\top=\mathbf{h}+\mathbf{\hat{w}}_{UT},$
with $\mathbf{\hat{w}}_{UT}\sim\mathcal{C}\mathcal{N}\left(\mathbf{0},\frac{N_0L}{\eta_FT_CP}\mathbf{I}_L\right)$ estimation error vector at the UT. The power consumption at UT to decode a pilot sequence sent from one of the $L$ transmit antennas is modeled similarly to the previous case, i.e., proportional to $P_D$. Accordingly, the total power consumed in decoding the pilot symbols is given by $\eta_F T_C P_D.$ 

At this stage, the UT encodes each coefficient by means of a sequence $\mathbf{f}_l=[f_l[1],\cdots,f_l[\tau_FT_C]]^\top$ $\in\mathbb{C}^{\tau_FT_C}$, $\forall l \in [1,L]$, such that the $L$ sequences form an orthogonal set, i.e., $\mathbf{f}_i \perp \mathbf{f}_j$, for $i \neq j$, and $\lVert\mathbf{f}_{1}\rVert^2=\cdots=\lVert\mathbf{f}_{L}\rVert^2=\frac{\tau_FT_CP_F}{L}$. As before, the adopted sequences are known at both ends of the communications. In particular, in order to guarantee their orthogonality and encode $L$ independent channel coefficients, a lower bound on the minimum sequence size must be satisfied, i.e., $\tau_FT_C\geq L$. 

After the encoding, the signal to be fed back by the UT to the AP is obtained as the sum of all the obtained sequences at the previous step, i.e., $\mathbf{x}^{f}_{F} = \mathbf{f}_{1}\hat{h}_{UT,1}+\cdots+\mathbf{f}_{L}\hat{h}_{UT,L}$. Consequently, its transmission requires a power given by
\begin{equation}
\frac{P_F}{L}\left(\lvert\hat{h}_{UT,1}\rvert^2+\cdots+\lvert\hat{h}_{UT,L}\rvert^2\right)=\frac{P_F\lVert\mathbf{\hat{h}}_{UT}\rVert^2}{L}.\label{Eq: FDD-Feedbacked power}
\end{equation}
Then, the received signal by the AP is given by 
$\mathbf{y}^{f}_{AP,F} =\mathbf{f}_{1}\hat{h}_{UT,1}+\cdots+\mathbf{f}_{L}\hat{h}_{UT,L}+\mathbf{w}_{AP},$
where $\mathbf{w}_{AP}\sim\mathcal{C}\mathcal{N}\left(\mathbf{0},N_0\mathbf{I}_{\tau_FT_C}\right)$ is the thermal noise at the AP. Now, the latter multiplies the received sequence by $\mathbf{f}_k^\dag/\lVert\mathbf{f}\rVert^2$ to estimate $h_k$. Thus, the estimated channel vector at the AP is obtained as
$\mathbf{\hat{h}}_{AP}=\mathbf{h}+\mathbf{\hat{w}}_{UT}+\mathbf{\hat{w}}_{AP},$
where $\mathbf{\hat{w}}_{AP}\sim\mathcal{C}\mathcal{N}\left(\mathbf{0},\frac{N_0L}{\tau_F T_CP_F}\mathbf{I}_L\right)$. In particular, we note that $\mathbf{\hat{w}}_{UT}$ and $\mathbf{\hat{w}}_{AP}$ are independent by definition.

Finally, the AP can exploit the knowledge of $\mathbf{\hat{h}}_{AP}$ to derive the desired MFP as before, given by $\mathbf{\hat{h}}_{AP}/\lVert\mathbf{\hat{h}}_{AP}\rVert$, and use it as beamforming vector while transmitting the information symbols for the remaining $(1-\alpha_F-\eta_F-\tau_F)T_C$ symbols.

The received information symbol at the UT is given by
\begin{equation}
y_{UT,F}^{s}=\frac{\mathbf{h}^{\dag}\mathbf{\hat{h}}_{AP}}{\lVert\mathbf{\hat{h}}_{AP}\rVert}s+n,\label{Eq: FDD-Received information}
\end{equation}
where $s$ is the information symbol, with $\mathbb{E}\left[\lvert s\rvert^2\right]=P$. Concerning the energy required to perform all the operations at the UT, as a matter of fact, since the harvested energy must be sufficient to decode the received pilot sequences, feedback the estimated CSI, and decode the subsequent information, we have that the condition
$\alpha_FT_C P_H=\eta_FT_C P_D+\tau_FT_CP_F\lVert\mathbf{\hat{h}}_{UT}\rVert^2/L+\left(1-\alpha_F-\eta_F-\tau_F\right)T_CP_D$
must be satisfied. Therefore, if we plug \eqref{Eq: System-Harvested power} into this condition then, after some manipulations, we have that the minimum duration of the energy transfer/harvesting phase for this case, i.e., $\alpha_F$, should be 
\begin{equation}\label{Eq: FDD-Time portion alpha}
\alpha_F=\frac{\tau_FP_F\lVert\mathbf{\hat{h}}_{UT}\rVert^2-\tau_FLP_D+LP_D}{\beta P\lVert\mathbf{h}\rVert^2+LP_D}.
\end{equation}
Now, we can use \eqref{Eq: FDD-Time portion alpha} to compute the downlink rate for the FDD scheme as
\begin{align}
R_F &=R_F(\alpha_F,\eta_F,\tau_F)\notag\\
&=\left(1-\alpha_F-\eta_F-\tau_F\right)\log_2\left(1+\frac{P\lvert\mathbf{h}^{\dag}\mathbf{\hat{h}}_{AP}\rvert^2}{N_0\lVert\mathbf{\hat{h}}_{AP}\rVert^2}\right)\notag
\\&=\frac{\left(1-\eta_F-\tau_F\right)\beta P\lVert\mathbf{h}\rVert^2-\tau_FP_F\lVert\mathbf{\hat{h}}_{UT}\rVert^2-\eta_FLP_D}{\beta P\lVert\mathbf{h}\rVert^2+LP_D}\notag
\\&\quad\times\log_2\left(1+\frac{P\lvert\mathbf{h}^{\dag}\mathbf{\hat{h}}_{AP}\rvert^2}{N_0\lVert\mathbf{\hat{h}}_{AP}\rVert^2}\right).\label{Eq: FDD-Rate}
\end{align}
In this case, two parameters describe the duration of the channel estimation phase, i.e., $\eta_F$ and $\tau_F$, related to the channel estimation procedures at the UT and the AP, respectively. In practice, these parameters impact both the accuracy of the estimated channel vectors and the remaining available time for information transfer at the AP. As a consequence, let $(\eta^{\star}_F, \tau^{\star}_F)$ be the optimal couple of parameters that maximizes the ergodic downlink rate, defined as
\begin{align}
&\hspace{-0.6em}(\eta^{\star}_F,\tau^{\star}_F) = \operatorname{argmax}_{\eta_F,\tau_F}\mathbb{E}
_{\mathbf{\mathbf{h},\hat{w}}}\Bigg[\log_2\left(1+\frac{P\lvert\mathbf{h}^{\dag}\mathbf{\hat{h}}_{AP}\rvert^2}{N_0\lVert\mathbf{\hat{h}}_{AP}\rVert^2}\right)\notag
\\&\hspace{-0.83em}\times \frac{\left(1-\eta_F-\tau_F\right)\beta P\lVert\mathbf{h}\rVert^2-\tau_FP_F\lVert\mathbf{\hat{h}}_{UT}\rVert^2-\eta_FLP_D}{\beta P\lVert\mathbf{h}\rVert^2+LP_D}\Bigg],\label{Eq: SecIII-FDD-rate optimal problem}
\end{align}
where $\mathbf{\hat{w}}=\left(\mathbf{\hat{w}}_{AP},\mathbf{\hat{w}}_{UT}\right)$. As before, the derivation of the exact value of $\eta^{\star}_F$ and $\tau^{\star}_F$ is very complicated. Nevertheless, two approximations of this value, valid for high and low SNR, respectively, can be derived as stated in the following result.
\begin{lemma}\label{Lem: FDD eta over all channels}
At high SNR, $\eta^{\star}_F$ and $\tau_F^{\star}$ can be approximated as
\begin{align}
\eta_F^{\star}&\approx\sqrt{\left(1+\frac{P_F}{\beta P}\right)\frac{P_F}{P}}\times\tau_F^{\star},\label{Eq: FDD_opt_high_SNR-1}
\\\tau_F^{\star}&\approx\sqrt{\frac{N_0L^2\log_2e}{B_5T_C\left(L-1\right)P_F\left(1+\frac{P_F}{\beta P}\right)}},\label{Eq: FDD_opt_high_SNR-2}
\end{align}
where $B_5=\mathbb{E}_\mathbf{h}\left[\log_2\left(\frac{P\lVert\mathbf{h}\rVert^2}{N_0}\right)\right]$.
At low SNR, they can be approximated as
\begin{align}
\tau^{\star}_F&\approx\frac{N_0L\left(-1+\sqrt{1+\frac{4\beta P^2\left(\frac{\beta P}{P_F}+1+\frac{N_0L}{\eta_F^{\star}T_CP}\right)}{\left(\frac{N_0L}{\eta_F^{\star}T_C}\right)^2}}\right)}{2T_C\left(\beta P+P_F+\frac{P_FN_0L}{\eta_F^{\star}T_CP}\right)},\label{Eq: FDD_opt_low_SNR-1}
\\\eta_F^{\star}&\approx\frac{P_FN_0L}{PT_C\left(\beta P+P_F\right)}
\left(-1+\sqrt{1+\frac{T_CP\left(\beta P+P_F\right)}{P_FN_0L}}\right).\label{Eq: FDD_opt_low_SNR-2}
\end{align}
\end{lemma}
\begin{proof}
See Appendix-\ref{Apx: FDD eta over all channels}.
\end{proof}
Despite the complexity of \eqref{Eq: FDD_opt_high_SNR-1}, \eqref{Eq: FDD_opt_high_SNR-2}, \eqref{Eq: FDD_opt_low_SNR-1} and \eqref{Eq: FDD_opt_low_SNR-2}, some interesting insights can drawn from the approximations adopted in the derivation in Appendix-\ref{Apx: FDD eta over all channels} following an approach similar to what has been done for Lemma \ref{Lem: TDD eta over all channels}. As before, the impact of $P_D$ on the accuracy of the results is likely negligible, due the fact that $P \gg P_D$ by construction. Now, consider the quantities $\lambda_1=\frac{2T_C\lVert\mathbf{h}\rVert^2}{N_0L\left(\frac{1}{\eta_F P}+\frac{1}{\tau_F P_F}\right)}$ and $\lambda_2=\frac{2\eta_F T_CP\lVert\mathbf{h}\rVert^2}{N_0L}$, introduced in \eqref{Eq: FDD eta high SNR neglect fraction-2} and \eqref{Eq: FDD eta low SNR Taylor series-2} respectively. We first focus on the low SNR case. Therein, the approximations $\lambda_1, \lambda_2 \approx 0$ are adopted. In this case, a smaller $P_F$ improves the accuracy of these approximations, whereas no clear insight can be drawn for $L$. Conversely, in the high SNR case, the approximation $\lambda_1 \gg 0$ is adopted. Differently from the previous case, the accuracy of this approximation increases with $P_F$. A further approximation is introduced in this part of the study, i.e., $\frac{N_0L^2}{\eta_FPT_C} \approx 0$ in \eqref{Eq: FDD eta high SNR neglect fraction-1}. Accordingly, an additional insight on the impact of the number of antennas on the accuracy of the result in Lemma 2 can be drawn, i.e., the smaller $L$ the larger the accuracy. Interestingly, this is in contrast with the impact of the same parameter in the TDD case and highlights the expected larger penalty for CSI acquisition that FDD pays w.r.t. TDD as the number of antennas grows. A more detailed discussion on its impact on the accuracy of the results in Lemma \ref{Lem: FDD eta over all channels}, is deferred to  Sec.~\ref{Sec: Numerical results}, where a comparative study of the downlink rate of the three considered schemes is provided.

\section{Analysis of the Outage Probability}\label{Sec: Outage probability}

In this section, we will study the outage probability for the considered system as a function of the parameters introduced so far, and the downlink rate. In the considered practical SWIPT implementation two possible outage events can occur:
\begin{itemize}
\item The harvested energy is not sufficient for all the operations at the UT (channel estimation, pilot transmission/CSI feedback and information decoding), i.e., the UT experiences an \textit{energy shortage}. 
\item The harvested energy is sufficient to perform all the operations at the UT, but the achieved downlink rate is smaller than a target value, i.e., the UT experiences a \textit{data outage}. 
\end{itemize}
We first focus on the case for which energy shortage occurs. Subsequently, we analyze the case for which the harvested energy is sufficient for all the operations at the UT, and compute the data outage probabilities for the three transmit schemes considered in this work. Before we proceed, we remark that, the analytic expressions derived in this section for the outage probabilities as a function of the system parameters are very complicated, and straightforward inference on their behavior is difficult to be drawn. Consequently, as before we defer the discussion on the outage as a function of the system parameters for all the cases considered in this work to Sec.~\ref{Sec: Numerical results}.

\subsection{Energy Shortage Probability} \label{Sec: Energy shortage}

\subsubsection{Non-CSI Scheme}

Referring to \eqref{Eq: Non-Time portion alpha}, for any given value for $\alpha_N$, the energy shortage probability for the non-CSI case can be expressed mathematically as
\begin{align}
\mathcal{P}_{N}^{E, out}\left(\alpha_N\right)&=\Pr\left\{\frac{\alpha_N\beta P\lVert\mathbf{h}\rVert^2}{L}<\left(1-\alpha_N\right)P_D\right\}\notag
\\&=\frac{\gamma\left(L,\frac{\left(1-\alpha_N\right)LP_D}{\alpha_N\beta P}\right)}{\Gamma\left(L\right)},\label{Eq: Non-Energy outage closed-form}
\end{align}
where $\gamma(q,r)=\int_0^{r}u^{q-1}e^{-u}\mathrm{d}u$ is the lower incomplete Gamma function. The closed-form expression of this probability is derived by considering the cumulative distribution function (CDF) of $\chi^2_{2L}$ if we note that $2\lVert\mathbf{h}\rVert^2\sim\chi^2_{2L}$.

\subsubsection{TDD Scheme}

Referring to \eqref{Eq: TDD-Time portion alpha}, for any given value for $\alpha_T$ and $\eta_T$, the energy shortage probability for the TDD case,  denoted by $\mathcal{P}_{T}^{E, out}\left(\alpha_T,\eta_T\right)$, can be expressed mathematically as
\begin{align}
&\mathcal{P}_{T}^{E, out}\left(\alpha_T,\eta_T\right)\notag
\\&\quad=\Pr\bigg\{\frac{\alpha_T\beta P\lVert\mathbf{h}\rVert^2}{L}<\left(1-\alpha_T-\eta_T\right)P_D+\eta_TP_E\bigg\}\notag
\\&\quad=\frac{\gamma\left(L,\frac{\eta_T LP_E+\left(1-\alpha_T-\eta_T\right)LP_D}{\alpha_T\beta P}\right)}{\Gamma\left(L\right)}.\label{Eq: TDD power outage probability}
\end{align}
Using the same approach as for \eqref{Eq: Non-Energy outage closed-form}, the closed-form expression of the probability in \eqref{Eq: TDD power outage probability} is computed.

\subsubsection{FDD Scheme}

Consider \eqref{Eq: FDD-Time portion alpha}. For any given value for $\alpha_F$, $\eta_F$ and $\tau_F$, the energy shortage probability for the FDD case can be stated mathematically as
\begin{align}
\mathcal{P}_{F}^{E, out}\left(\alpha_F,\eta_F,\tau_F\right)=\Pr\bigg\{
&\frac{\alpha_F\beta P\lVert\mathbf{h}\rVert^2}{L}<\frac{\tau_FP_F\lVert\mathbf{\hat{h}}_{UT}\rVert^2}{L}
\notag
\\&+\left(1-\alpha_F-\tau_F\right)P_D\bigg\}.\label{Eq: FDD-Energy outage-1}
\end{align}
The following result provides a closed-form expression of \eqref{Eq: FDD-Energy outage-1}. However, for the sake of the simplicity of the representation of the result, let us denote $\sigma_1=\sqrt{\frac{N_0L}{2\eta_F T_CP}}$, $\rho_1=\frac{\sqrt{2}\tau_F P_F}{\alpha_F\beta P-\tau_F P_F}$,
$\rho_2=\sqrt{\frac{2\tau_F P_F}{\alpha_F\beta P-\tau_F P_F}+\frac{2\tau_F^2P_F^2}{\left(\alpha_F\beta P-\tau_F P_F\right)^2}}$, and
$\rho_3=\sqrt{\frac{2\left(1-\alpha_F-\tau_F\right)LP_D}{\alpha_F\beta P-\tau_F P_F}}$.
\begin{lemma}\label{Lem: FDD power outage probability}
The energy shortage probability for the FDD scheme, as in \eqref{Eq: FDD-Energy outage-1}, can be computed as
\begin{align}
&\mathcal{P}_{F}^{E, out}\left(\alpha_F,\eta_F,\tau_F\right)\notag
\\&\quad=1-\int_{\theta_4 = 0}^{\infty}\frac{Q_{L}\left(\sqrt{\rho_1^2\sigma_1^2\theta_4},\sqrt{\rho_2^2\sigma_1^2\theta_4+\rho_3^2}\right)}{e^{\frac{\theta_4}{2}}\theta_4^{1-L}2^{L}\Gamma\left(L\right)}
\mathrm{d}\theta_4.\label{Eq: FDD power outage closed form}
\end{align}
\end{lemma}
\begin{proof}
The outage probability can be evaluated as follows. First, applying the law of total probability, i.e., given a random variable $A$,
$\Pr\left(\cdot\right)=\mathbb{E}_{A}\left[\Pr\left(\cdot|A\right)\right]$, we have
\begin{align}
&\eqref{Eq: FDD-Energy outage-1}=\mathbb{E}_{\mathbf{\hat{w}}_{UT}}\Big[\Pr\Big\{
\lVert\sqrt{2}\mathbf{h}-\rho_1\mathbf{\hat{w}}_{UT}\rVert^2\notag
\\&\hspace{10em}<
\rho_2^2\lVert\mathbf{\hat{w}}_{UT}\rVert^2
+\rho_3^2\big|\mathbf{\hat{w}}_{UT}
\Big\}\Big].\label{Eq: App4-Outage probability-1}
\end{align}
From \eqref{Eq: App4-Outage probability-1}, it can be easily deduced that $\lVert\sqrt{2}\mathbf{h}-\rho_1\mathbf{\hat{w}}_{UT}\rVert^2\big|_{\mathbf{\hat{w}}_{UT}}\sim\chi^{'2}_{2L}\left(\rho_1^2\lVert\mathbf{\hat{w}}_{UT}\rVert^2\right)$. Therefore, substituting the PDF of $\lVert\sqrt{2}\mathbf{h}-\rho_1\mathbf{\hat{w}}_{UT}\rVert^2\big|_{\mathbf{\hat{w}}_{UT}}$ into \eqref{Eq: App4-Outage probability-1}, we can rewrite 
\begin{equation}
\eqref{Eq: App4-Outage probability-1} = 1-\mathbb{E}_{\mathbf{\hat{w}}_{UT}}\left[Q_{L}\left(\rho_1\lVert\mathbf{\hat{w}}_{UT}\rVert,\sqrt{\rho_2^2\lVert\mathbf{\hat{w}}_{UT}\rVert^2+\rho_3^2}\right)\right].\label{Eq: App4-Outage probability-2}
\end{equation}
Since $\mathbf{\hat{w}}_{UT}\sim\mathcal{C}\mathcal{N}\left(\mathbf{0},2\sigma_1^2\mathbf{I}_L\right)$, we have $\lVert\mathbf{\hat{w}}_{UT}\rVert^2=\sigma_1^2\Theta_4$, where $\Theta_4\sim\chi^2_{2L}$. Substituting the PDF of $\lVert\mathbf{\hat{w}}_{UT}\rVert^2$ into \eqref{Eq: App4-Outage probability-2}, we derive the RHS of \eqref{Eq: FDD power outage closed form}, and this concludes the proof.

\end{proof}
At this stage, if we focus on \eqref{Eq: Non-Energy outage closed-form}, \eqref{Eq: TDD power outage probability}, and \eqref{Eq: FDD power outage closed form}, we note that the energy shortage probability in the three considered cases clearly depends on the values of $\alpha_N$, $(\alpha_T,\eta_T)$, and $(\alpha_F,\eta_F,\tau_F)$ respectively. However, drawing meaningful insights from these results is extremely difficult, due to their complexity. Accordingly, we will investigate this aspect in Sec.~\ref{Sec: Numerical results}, by means of suitable numerical analyses.

\subsection{Data Outage Probability for the Non-CSI Scheme}\label{sec:data outafe for non-CSI}

We now compute the data outage probability for the non-CSI scheme. Given $\alpha_N$ and a specific target downlink rate $R_{NC}$, the data outage probability can be stated mathematically as
\begin{align}
&\mathcal{P}_{N}^{D, out}\left(\alpha_N,R_{NC}\right)=\Pr\bigg\{\frac{\alpha_N\beta P\lVert\mathbf{h}\rVert^2}{L}\geq\left(1-\alpha_N\right)P_D,\notag
\\&\hspace{4em}\left(1-\alpha_N\right)\log_2\left(1+\frac{P\lVert\mathbf{h}\rVert^2}{N_0L}\right)<R_{NC}\bigg\},\notag
\end{align}
that is the probability that the harvested energy is sufficient for the decoding operations at the UT, but the achieved downlink rate is smaller than $R_{NC}$. Now, let us rewrite $\mathcal{P}_{N}^{D, out}$ as
\begin{align}
&\hspace{-0.7em}\Pr\left\{\frac{\left(1-\alpha_N\right)LP_D}{\alpha_N\beta P}\leq
\lVert\mathbf{h}\rVert^2<\frac{N_0L}{P}\left(2^{\frac{R_{NC}}{1-\alpha_N}}-1\right)
\right\}.\label{Eq: Non-Data outage closed-form-1}
\end{align}
The intersection between the two events in \eqref{Eq: Non-Data outage closed-form-1} is non-empty when
\begin{equation} \label{eq:condition_non_CSI}
\frac{\left(1-\alpha_N\right)LP_D}{\alpha_N\beta P}<\frac{N_0L}{P}\left(2^{\frac{R_{NC}}{1-\alpha_N}}-1\right).
\end{equation}
If this condition is not satisfied, then \eqref{Eq: Non-Data outage closed-form-1} in this case is equal to 0. In is worth noting that, assuming $R_{NC}\neq0$, the data outage probability would be 0 only in case of extremely low value of $N_0$, given that typically $P \gg P_D$, as previously discussed. This is in line with what could be expected in a wireless communication system, in which the data outage probability tends to 0 as the SNR at the receiver increases. If this is not the case, and \eqref{eq:condition_non_CSI} is satisfied, then \eqref{Eq: Non-Data outage closed-form-1} can be computed as
\begin{align}
&\mathcal{P}_{N}^{D, out}\left(\alpha_N,R_{NC}\right)\notag
\\&=\frac{\gamma\left(L,\frac{N_0L}{P}\left(2^{\frac{R_{NC}}{1-\alpha_N}}-1\right)\right)}{\Gamma\left(L\right)}-\frac{\gamma\left(L,\frac{\left(1-\alpha_N\right)LP_D}{\alpha_N\beta P}\right)}{\Gamma\left(L\right)},\label{Eq: Non-Data outage closed-form-2}
\end{align}
where we made use of the CDF of the $\chi^2_{2L}$ distribution.

\subsection{Data Outage Probability for the TDD Scheme}\label{sec:data outafe for TDD}

We switch our focus back to the TDD scheme. For given values of $\alpha_T$, $\eta_T$, and a target downlink rate $R_T$, the data outage probability is expressed as
\begin{align}
&\hspace{-0.7em}\mathcal{P}_{T}^{D, out}\left(\alpha_T,\eta_T,R_{T}\right)=\Pr\bigg\{\frac{\alpha_T\beta P\lVert\mathbf{h}\rVert^2}{L}\geq\left(1-\alpha_T-\eta_T\right)P_D\notag
\\&\hspace{-0.7em}+\eta_TP_E,\left(1-\alpha_T-\eta_T\right)\log_2\left(1+\frac{P\lvert\mathbf{h}^{\dag}\mathbf{\hat{h}}\rvert^2}{N_0\lVert\mathbf{\hat{h}}\rVert^2}\right)<R_T
\bigg\},\label{Eq: TDD-Data outage}
\end{align}
that is the probability that the harvested energy is sufficient to engage in the pilots transmission and decode the received data, but the achieved downlink rate is smaller than $R_T$. 
The following result provides a closed-form expression for \eqref{Eq: TDD-Data outage} and concludes the study of the TDD case. However, before proceeding, let us denote $b_3=\frac{N_0}{P}\left(2^{\frac{R_T}{1-\alpha_T-\eta_T}}-1\right)$, $b_4=\frac{\eta_T LP_E+\left(1-\alpha_T-\eta_T\right)LP_D}{\alpha_T\beta P}$, $b_5=\frac{N_0+\eta_T T_CP_E}{N_0}$ and $b_6=\frac{N_0}{\eta_T T_CP_E}$, for the sake of the simplicity of the representation of the result.
\begin{lemma}\label{Lem: TDD data outage probability}
When $\frac{N_0}{P}\left(2^{\frac{R_T}{1-\alpha_T-\eta_T}}-1\right)<\frac{\eta_T LP_E+\left(1-\alpha_T-\eta_T\right)LP_D}{\alpha_T\beta P}$,
then the data outage probability for the TDD scheme, as in \eqref{Eq: TDD-Data outage}, can be computed as
\begin{align}
&\mathcal{P}_{T}^{D, out}=\int_{\theta_3 = 0}^{\infty}\int_{\theta_1 = 0}^{2b_5 b_3}\frac{\Gamma\left(L-1,b_5 b_4-\frac{\theta_1}{2}\right)
\theta_3^{L-1}}{2^{L+1}\Gamma\left(L-1\right)\Gamma\left(L\right)}\notag
\\&\qquad\times I_0\left(\sqrt{\frac{\theta_1 \theta_3}{b_6}}\right)e^{-\left(\frac{\theta_1}{2}+\frac{\theta_3}{2b_6}+\frac{\theta_3}{2}\right)}\mathrm{d}\theta_1\mathrm{d}\theta_3.\label{Eq: TDD data outage probability closed form 1}
\end{align}
Conversely, when $\frac{N_0}{P}\left(2^{\frac{R_T}{1-\alpha_T-\eta_T}}-1\right)\geq\frac{\eta_T LP_E+\left(1-\alpha_T-\eta_T\right)LP_D}{\alpha_T\beta P}$, it can be computed as
\begin{align}
&\mathcal{P}_{T}^{D, out}=\int_{\theta_3 = 0}^{\infty}\int_{\theta_1 = 0}^{2b_5 b_4}\frac{\Gamma\left(L-1,b_5 b_4-\frac{\theta_1}{2}\right)
I_0\left(\sqrt{\frac{\theta_1 \theta_3}{b_6}}\right)}{2^{L+1}\Gamma\left(L-1\right)\Gamma\left(L\right)}\notag
\\&\times \theta_3^{L-1}e^{-\left(\frac{\theta_1}{2}+\frac{\theta_3}{2b_6}+\frac{\theta_3}{2}\right)}\mathrm{d}\theta_1\mathrm{d}\theta_3+\int_{\theta_3 = 0}^{\infty}e^{-\frac{\theta_3}{2}}\theta_3^{L-1}\notag
\\&\times
\frac{\left(Q_{1}\left(\sqrt{\frac{\theta_3}{b_6}},\sqrt{2b_5 b_4}\right)
-Q_{1}\left(\sqrt{\frac{\theta_3}{b_6}},\sqrt{2b_5 b_3}\right)\right)}{2^L\Gamma\left(L\right)}\mathrm{d}\theta_3.\label{Eq: TDD data outage probability closed form 2}
\end{align}
\end{lemma}
\begin{proof}
See Appendix-\ref{Apx: TDD data outage probability}.
\end{proof}

\subsection{Data Outage Probability for the FDD Scheme}\label{sec:data outafe for FDD}
We conclude our study on the data outage probability by considering the FDD case. For given $\alpha_F$, $\eta_F$, $\tau_F$, and a specific target downlink rate $R_F$, the data outage probability can be stated mathematically as
\begin{align}
&\hspace{-0.5em}\mathcal{P}_{F}^{D, out}\left(\alpha_F,\eta_F,\tau_F,R_{F}\right)\notag
\\&\hspace{-0.5em}=\Pr\Bigg\{\frac{\alpha_F\beta P\lVert\mathbf{h}\rVert^2}{L}\geq\frac{\tau_FP_F\lVert\mathbf{\hat{h}}_{UT}\rVert^2}{L}
+\left(1-\alpha_F-\tau_F\right)P_D,\notag
\\&\hspace{-0.5em}\left(1-\alpha_F-\eta_F-\tau_F\right)\log_2\left(1+\frac{P\lvert\mathbf{h}^{\dag}\mathbf{\hat{h}}_{AP}\rvert^2}{N_0\lVert\mathbf{\hat{h}}_{AP}\rVert^2}\right)<R_F\Bigg\},\label{Eq: FDD-Data outage-1}
\end{align}
that is the probability that the harvested energy is sufficient to estimate the downlink channel, feed back its estimated version in the uplink, and decode the received data, but the achieved downlink rate is smaller than $R_F$. 
The following result provides a closed-form expression of \eqref{Eq: FDD-Data outage-1} and concludes the study of the FDD case. However, as before, let us introduce some new notation to further simplify representation of the results. Accordingly, we let $\sigma_2=\frac{N_0L+\eta_F PT_C}{N_0L}$, $\sigma_3=\frac{N_0L}{\eta_F PT_C}$, $\sigma_4=\frac{N_0L}{\tau_F P_FT_C}$, $\sigma_5=\frac{\left(1+\sigma_3\right)\sigma_4}{1+\sigma_3+\sigma_4}$, $b_7=\frac{N_0}{P}\left(2^{\frac{R_F}{1-\alpha_F-\eta_F-\tau_F}}-1\right)$, $b_8=\frac{\left(1-\alpha_F-\tau_F\right)LP_D}{\alpha_F\beta P}$, and $b_9=\frac{\tau_F P_F}{\alpha_F\beta P}$.
\begin{lemma}\label{Lem: FDD data outage probability}

The data outage probability for the FDD scheme, as in \eqref{Eq: FDD-Data outage-1}, can be computed as
\begin{align}
&\mathcal{P}_{F}^{D, out}=\int_{\theta_9 = 0}^{\infty}\int_{\theta_7+\theta_8>\frac{2\left(b_7-b_8\right)}{b_9\sigma_5}}\int_{\theta_5 = 0}^{2\sigma_2b_7}\frac{\theta_8^{L-2}\theta_9^{L-1}}{\Gamma\left(L\right)\Gamma\left(L-1\right)}\notag
\\&\times Q_{L-1}\left(\sqrt{\frac{\theta_8\sigma_5}{\sigma_2\sigma_3^2}},
\sqrt{2\sigma_2\left(b_8+\frac{b_9\left(\theta_7+\theta_8\right)\sigma_5}{2}\right)-\theta_5}\right)\notag
\\&\times\frac{I_0\left(\sqrt{\frac{\theta_5\theta_7\sigma_5}{\sigma_2\sigma_3^2}}\right)I_0\left(\sqrt{\frac{\theta_7\theta_9\left(1+\sigma_3\right)}{\sigma_4}}\right)}{2^{2L+1}\times e^{\left(\frac{\theta_7\sigma_5}{2\sigma_2\sigma_3^2}+\frac{\theta_5+\theta_7+\theta_8+\theta_9}{2}+\frac{\left(1+\sigma_3\right)z}{2\sigma_4}\right)}}\mathrm{d}\theta_5\mathrm{d}\theta_7\mathrm{d}\theta_8\mathrm{d}\theta_9\label{Eq: FDD-Data outage closed form-1}
\\&+\int_{\theta_9 = 0}^{\infty}\int_{\theta_7+\theta_8\leq\frac{2\left(b_7-b_8\right)}{b_9\sigma_5}}\int_{\theta_5 = 0}^{2\sigma_2 \left(b_8+\frac{b_9\left(\theta_7+\theta_8\right)\sigma_5}{2}\right)}\notag
\\&\quad Q_{L-1}\left(\sqrt{\frac{\theta_8\sigma_5}{\sigma_2\sigma_3^2}},
\sqrt{2\sigma_2\left(b_8+\frac{b_9\left(\theta_7+\theta_8\right)\sigma_5}{2}\right)-\theta_5}\right)\notag
\\&\times I_0\left(\sqrt{\frac{\theta_5\theta_7\sigma_5}{\sigma_2\sigma_3^2}}\right)
I_0\left(\sqrt{\frac{\theta_7\theta_9\left(1+\sigma_3\right)}{\sigma_4}}\right)\theta_8^{L-2}\theta_9^{L-1}\notag
\\&\times\frac{e^{-\left(\frac{\theta_7\sigma_5}{2\sigma_2\sigma_3^2}+\frac{\theta_5+\theta_7+\theta_8+\theta_9}{2}+\frac{\left(1+\sigma_3\right)\theta_9}{2\sigma_4}\right)}}{\Gamma\left(L-1\right)\Gamma\left(L\right)2^{2L+1}}\mathrm{d}\theta_5\mathrm{d}\theta_7\mathrm{d}\theta_8\mathrm{d}\theta_9\label{Eq: FDD-Data outage closed form-2}
\\&+\int_{\theta_9 = 0}^{\infty}\int_{\theta_7+\theta_8\leq\frac{2\left(b_7-b_8\right)}{b_9\sigma_5}}
\frac{I_0\left(\sqrt{\frac{\theta_7\theta_9\left(1+\sigma_3\right)}{\sigma_4}}\right)}{2^{2L}\times e^{\left(\frac{\left(1+\sigma_3\right)\theta_9}{2\sigma_4}+\frac{\theta_7+\theta_8+\theta_9}{2}\right)}}\notag
\\&\times\Bigg[Q_1\left(\sqrt{\frac{\theta_7\sigma_5}{\sigma_2\sigma_3^2}},\sqrt{2\sigma_2 \left(b_8+\frac{b_9\left(\theta_7+\theta_8\right)\sigma_5}{2}\right)}\right)\notag
\end{align}
\begin{align}&-Q_1\left(\sqrt{\frac{\theta_7\sigma_5}{\sigma_2\sigma_3^2}},\sqrt{2\sigma_2 b_7}\right)\Bigg]
\times \frac{\theta_8^{L-2}\theta_9^{L-1}}{\Gamma\left(L-1\right)\Gamma\left(L\right)}\mathrm{d}\theta_7\mathrm{d}\theta_8\mathrm{d}\theta_9.\label{Eq: FDD-Data outage closed form-3}
\end{align}
\end{lemma}
\begin{proof}
See Appendix-\ref{Apx: FDD data outage probability}.
\end{proof}

\section{Numerical Results}\label{Sec: Numerical results}

In this section we evaluate the performance of SWIPT for MISO systems, to assess its merit under the transmit schemes considered in this work. The parameters used in our numerical results are as follows. We consider $\beta=0.5$, which is a good approximation of the performance delivered by state-of-the-art commercial products \cite{prod:P2110}, and typically adopted value in the literature on this subject \cite{ art:zhang13, art:kaibin13, art:liu13}. We consider $P=1$ and $T_C=1000$ for simplicity, and $L\in\{3,6\}$. We assume that the system operates in the industrial, scientific and medical (ISM) band, i.e., carrier frequency of 2.4~GHz. Accordingly, we set a distance between the AP and the UT in the order of meters such that we can ensure that the latter is situated in the far-field region of the radiating AP. Furthermore, we assume that the signals transmitted by both the AP and the UT experience a generic path loss attenuation, with a path loss exponent equal to 3. As a consequence, we can safely let $\frac{P}{P_D}=1000$. The rationale for this is that by incorporating the propagation losses in $\frac{P}{P_D}$, we frame a more realistic scenario. Finally, we model the ratio between the power budgets available at the AP and the UT following the same logic. We let $\frac{P}{P_E}=\frac{P}{P_F}=100$, in accordance with the typical ratio between the available power budgets at both sides of the communication in modern networks, which is roughly $20$~dB \cite{rpt:3gpp.36.814}.

\subsection{Downlink Rate}

Now, we focus on the ergodic downlink rate. First, we compute the optimal numerical performance of the system by numerically solving the problems in \eqref{Eq: SecIII-TDD-rate optimal problem} and \eqref{Eq: SecIII-FDD-rate optimal problem}, by means of an exhaustive search whose complexity and time requirements are not suitable for realistic implementations. Subsequently, we evaluate the accuracy of our theoretical results by comparing them to the numerical performance results. Throughout this section, we will refer to the derived approximated parameters in Lemma \ref{Lem: TDD eta over all channels} and Lemma \ref{Lem: FDD eta over all channels} as analytic results, for the sake of clarity. Now, for the TDD scheme,  let $R^\star_T$ and $\eta^\star_T$ be the optimal downlink rate and the optimal duration of the portion of the coherence time devoted to the channel training/estimation, computed by extensive Monte-Carlo simulations. For the sake of clarity, with a little abuse of notation, we denote $\hat{\eta}^\star_T$ as the optimal parameter of interest for the TDD scheme, computed according to Lemma \ref{Lem: TDD eta over all channels}. For the FDD scheme, a similar notation is defined. Now, we define
$\zeta_T = \frac{R_T (\alpha_T, \hat{\eta}^\star_T)}{R^\star_T} \in [0,1]$, for TDD, and
$\zeta_F = \frac{R_F (\alpha_F,\hat{\eta}^\star_F,\hat{\tau}^\star_F)}{R^\star_F} \in [0,1]$, for FDD,
as the ratio between the downlink rate obtained with the analytic and optimal numerical results.\footnote{Note that, $\alpha_T$ and $\alpha_F$
%in \eqref{eq:zeta} 
are computed according to \eqref{Eq: TDD-Time portion alpha} and \eqref{Eq: FDD-Time portion alpha} respectively.} We let SNR $\in [0,30]$~dB and compute $\zeta_T$ and $\zeta_F$ for both $L=3$ and $L=6$ in Fig.~\ref{fig:TDD_L3}, Fig.~\ref{fig:TDD_L6}, Fig.~\ref{fig:FDD_L3}, and Fig.~\ref{fig:FDD_L6}.
\begin{figure}[!h]
	\centering
	\includegraphics[width=\columnwidth]{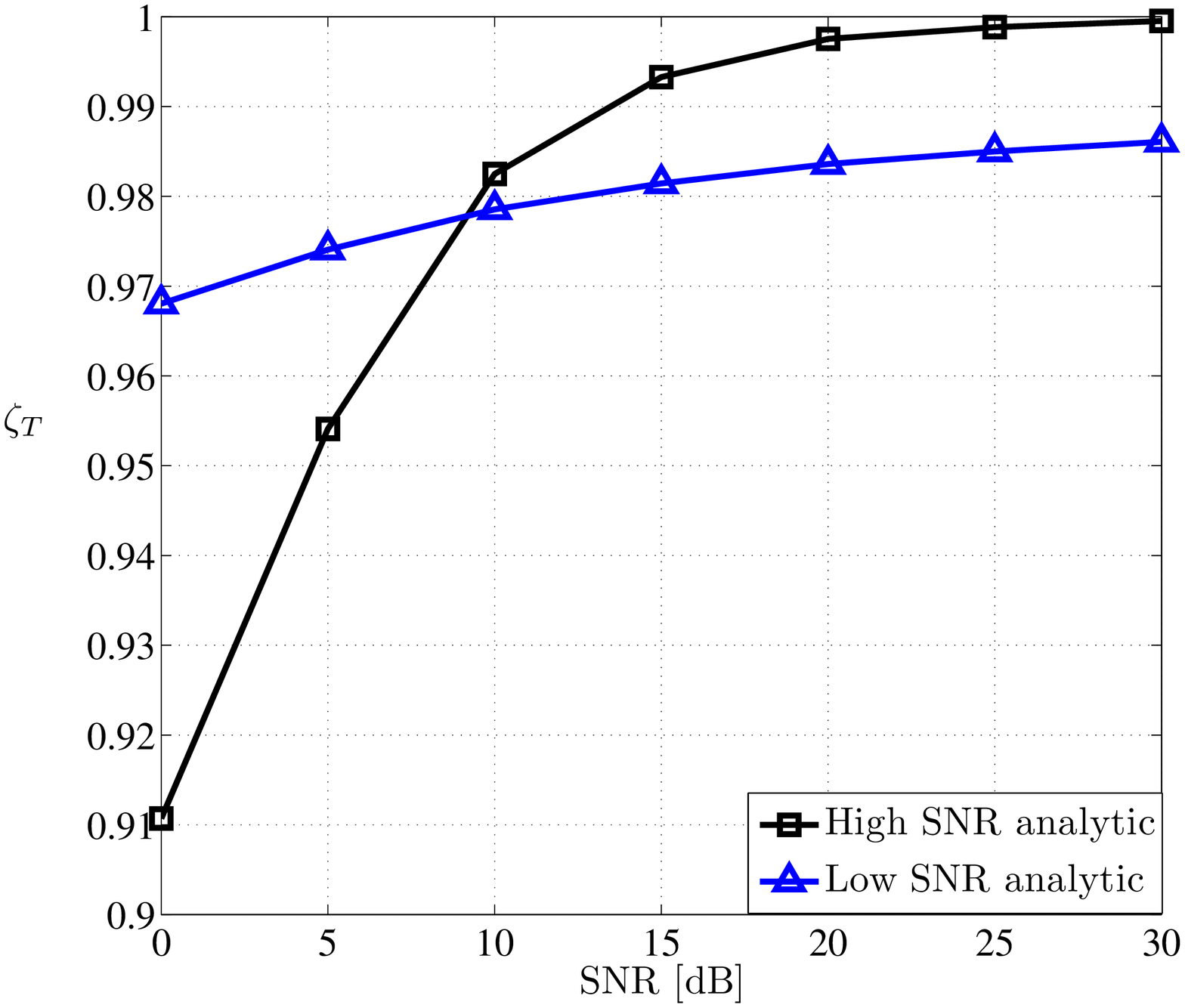}
	\caption{$\zeta_T$ for analytic and numerical parameters, TDD and $L=3$ antennas.}	
	\label{fig:TDD_L3}
\end{figure}
\begin{figure}[!h]
	\centering
	\includegraphics[width=\columnwidth]{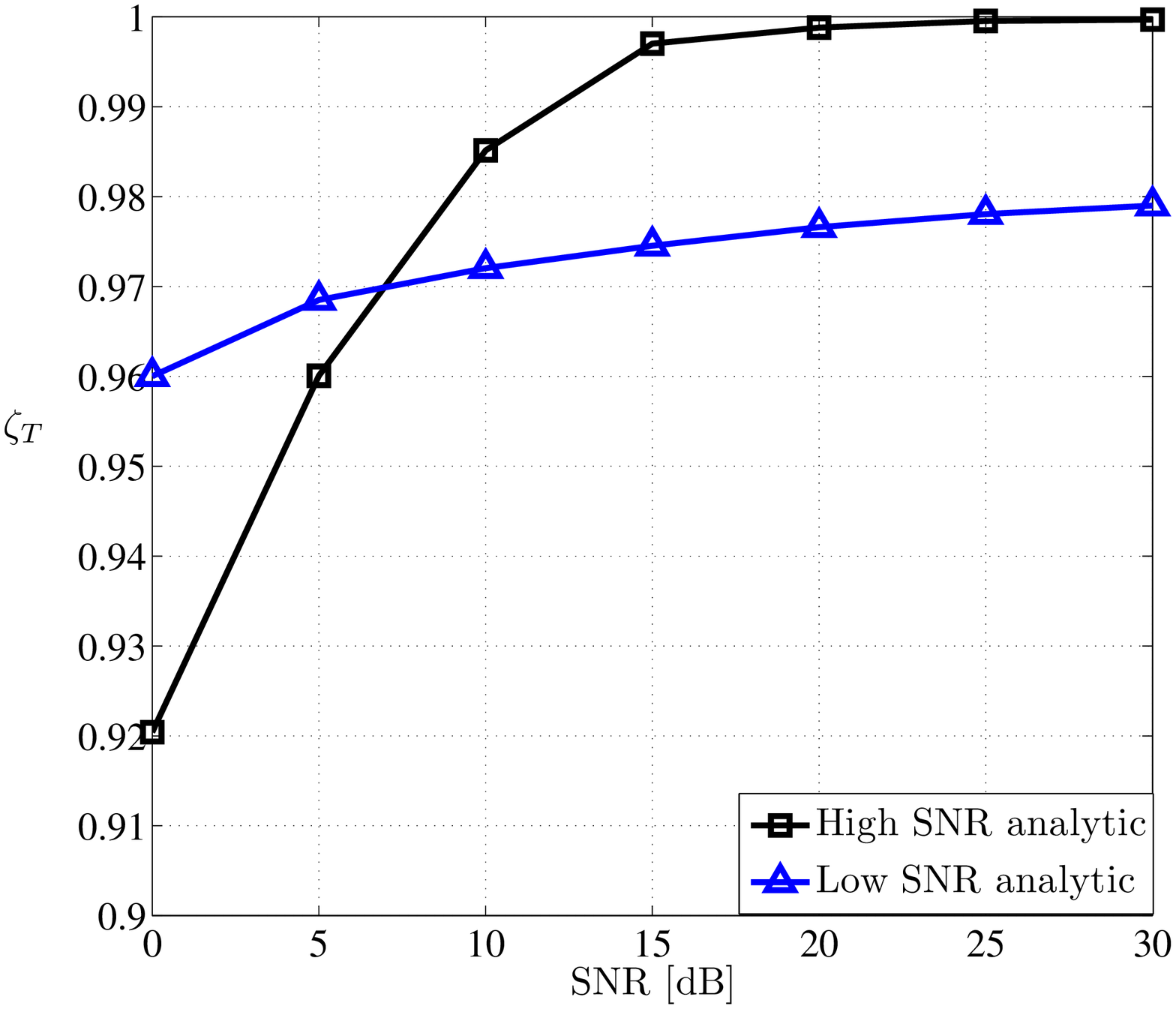}
	\caption{$\zeta_T$ for analytic and numerical parameters, TDD and $L=6$ antennas.}	
	\label{fig:TDD_L6}
\end{figure}
\begin{figure}[!h]
	\centering
	\includegraphics[width=\columnwidth]{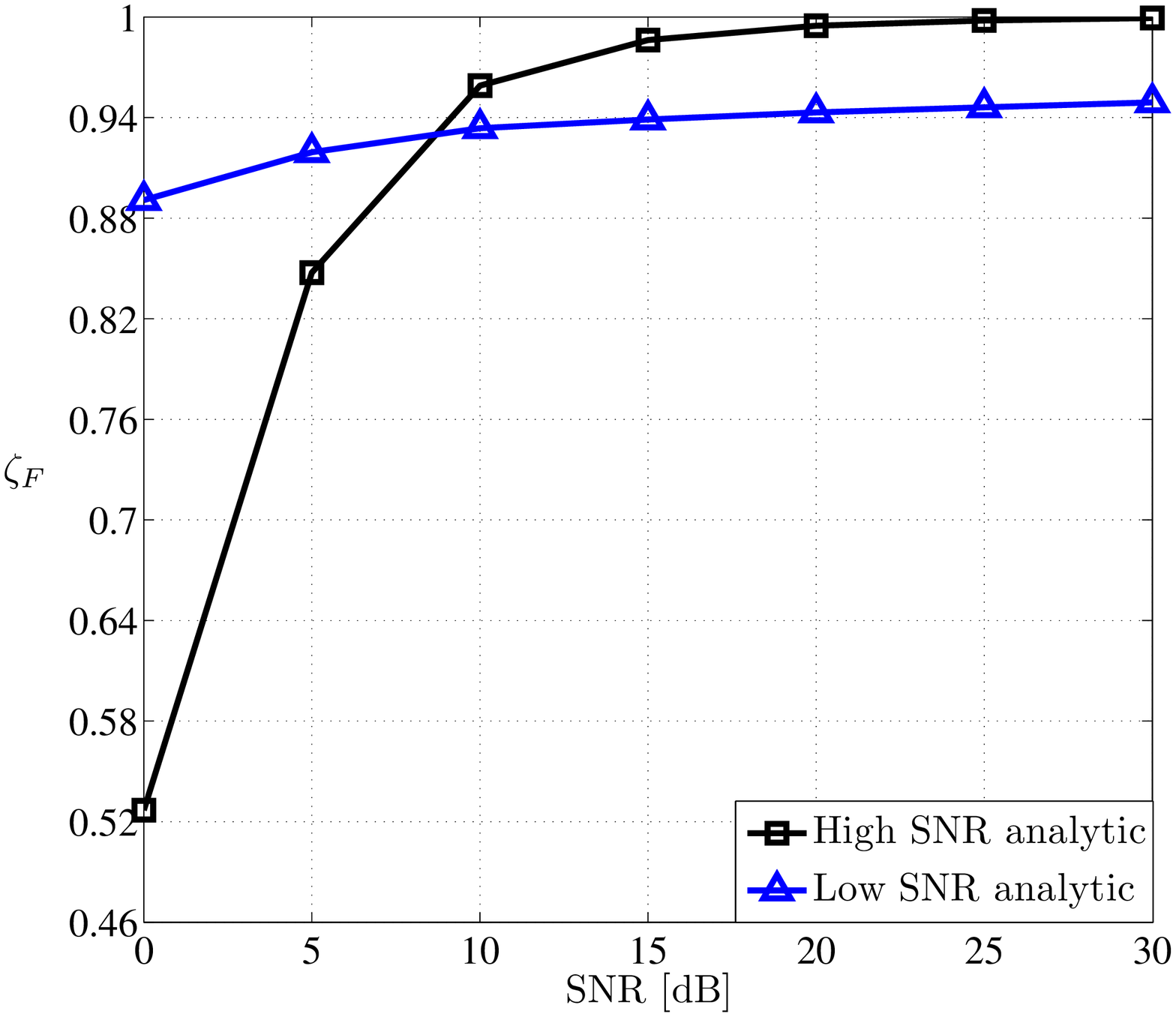}
	\caption{$\zeta_F$ for analytic and numerical parameters, FDD and $L=3$ antennas.}	
	\label{fig:FDD_L3}
\end{figure}
\begin{figure}[!h]
	\centering
	\includegraphics[width=\columnwidth]{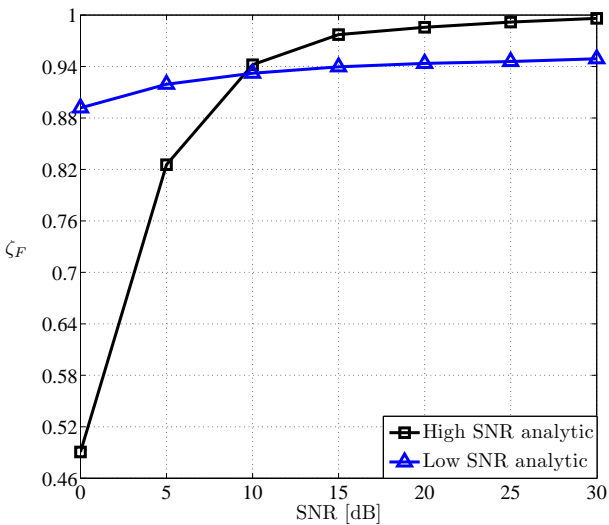}
	\caption{$\zeta_F$ for analytic and numerical parameters, FDD and $L=6$ antennas.}	
	\label{fig:FDD_L6}
\end{figure}
%
%
%\noindent
Quantitatively, if we focus on the best performer for each of the considered SNR values, the gap between $\zeta_T$ ($\zeta_F$ in the FDD case) and 1 is remarkably small. Thus, the accuracy of our derivations is confirmed.
If we focus on the impact of $L$ on the two analytic results, we note that they confirm the intuitions provided in Sec.~\ref{Sec: TDD scheme} and Sec.~\ref{Sec: FDD scheme}. However, the difference in terms of the best $\zeta_T$ (and $\zeta_F$) between the two antenna configurations is rather small. This shows that the impact of the number of antennas at the AP on the accuracy of the analytic results is not very significant.
Furthermore, we see that $\zeta_F\leq\zeta_T$, $\forall$ SNR$\in [0,30]$~dB and $\forall\,L \in\{3,6\}$. This is due to the two-step channel estimation process that is needed in the FDD scheme for the CSI acquisition at the AP. As a consequence,  a greater number of approximations is necessary. This reduces the accuracy of our closed-form representation of $\eta^\star_F$ and $\tau^\star_F$.\footnote{The interested reader may refer to Appendix-\ref{Apx: FDD eta over all channels} for further details.}

Focusing on the practical implementation, we note that the presence of the analytic results provides a twofold alternative for the AP, depending on the system intrinsic constraints. When the time available for the optimization of the transmit parameters is small, the analytic results could be used to achieve a performance which is reasonably close to the optimal, without resorting to an exhaustive search. Conversely, if more time is available for the AP, the analytic results can be used to improve the efficiency of the search for the optimal parameters. In this regard, we note that the downlink rate is a concave function of $(\eta,\tau)$. Accordingly, in this case, the local optimum coincides with the global optimum. Now, assume that the results of Lemma \ref{Lem: TDD eta over all channels} and Lemma \ref{Lem: FDD eta over all channels} were adopted as a starting point for finding the numerically optimal parameters, by means of an exhaustive search inside a smaller set. Then, the necessary time to identify the global optimum could be significantly reduced w.r.t. a ``blind'' exhaustive search, due to the proximity of the analytic results and the actual global optimum.

To conclude our analysis on the downlink rate, we investigate the advantages, if any, that the two duplexing schemes discussed so far can bring w.r.t. the non-CSI case in terms of the downlink rate. We remark that our goal is to characterize the performance of the system under the realistic assumptions made in Sec.~\ref{sec:introduction}. Thus, in the following study, all the system parameters discussed so far are set according to the analytic results derived in Sec.~\ref{sec:downlink}. Moreover, for simplicity in the representation, we let 
$R=R_T (\alpha_T,\hat{\eta}^\star_T)$ and
$R=R_F (\alpha_F,\hat{\eta}^\star_F,\hat{\tau}^\star_F)$
be the downlink rate for TDD and FDD, respectively, when the analytic results are adopted. The ratio between these rates and their counterpart for the non-CSI case (i.e., $\frac{R}{R_{NC}}$, with $R_{NC}$ as in \eqref{Eq: Non-Rate}) is represented in Fig.~\ref{fig:ergodic}, for SNR$\in [0,30]$~dB.
\begin{figure}[!h]
	\centering
	\includegraphics[width=\columnwidth]{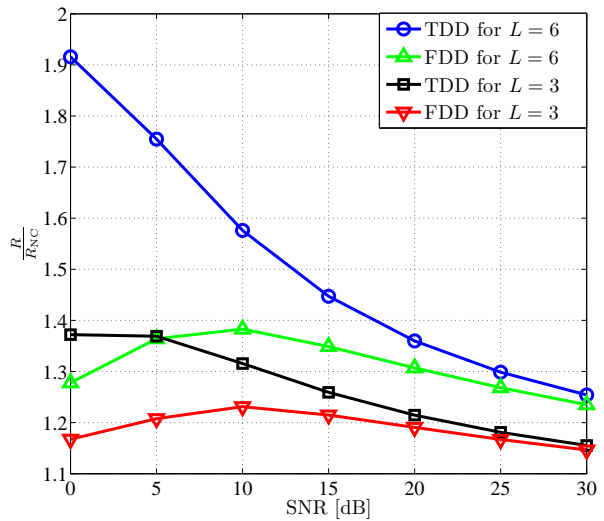}
	\caption{Ratio between the ergodic downlink rate for the CSI acquisition schemes and the non-CSI case.}	
	\label{fig:ergodic}
\end{figure}
Remarkably, both duplexing schemes clearly outperform the non-CSI approach in terms of downlink rate. This shows that, despite the penalties incurred to acquire the CSI, evident downlink rate enhancements are experienced by the AP, thanks to presence of the CSI, however imperfect the latter might be. 
The result in Fig.~\ref{fig:ergodic} is even more remarkable, considering that therein the two duplexing schemes always outperform the non-CSI approach, regardless of the antenna configuration and the SNR value. Furthermore, the largest advantage over the non-CSI performance is obtained in the low-to-mid SNR regime. In this regard, we first focus on the TDD case. In both cases, i.e., $L=3$ and $L=6$, $\frac{R}{R_{NC}}$ is a monotonically decreasing function of the SNR, confirming that the MFP performs better for low than for high SNR values \cite{conf:demig2008, TseWireless}. In particular, this shows that the availability of the CSI at the AP, albeit imperfect, is sufficient to achieve a much larger downlink rate as compared to the non-CSI approach. Furthermore, the performance for $L=6$ is strictly larger than for $L=3$, showing that, as in the case of traditional wireless communications, the SWIPT can effectively exploit the transmit diversity gain delivered by a MISO system as $L$ grows. Interestingly, the same is true for the FDD scheme. The CSI acquisition procedure in this case is more complex and prone to a higher uncertainty, especially at low SNR. This impacts the behavior of $\frac{R}{R_{NC}}$ that presents a maximum at SNR$=10$~dB, for both the considered antenna configurations. On one hand, the gain brought by the FDD scheme over the non-CSI approach is dominated by the power gain at the UT, brought by a more accurate beamformer design at the AP, for SNR$\leq10$~dB. On the other hand, the reduction of the multiplexing gain due to the increasing impact that both the channel estimation and feedback phases have on available time for information transfer, as the quality of the CSI increases, determines the decreasing behavior of $\frac{R}{R_{NC}}$ for SNR$>10$~dB. Finally, we note that the difference at high SNR between the values of $\frac{R}{R_{NC}}$ for TDD and FDD is very low, but increases with the $L$. In fact, when the SNR is high, the channel estimation/feedback phases are very short, thus the difference in the amount of time available for the information transfer in both cases is small. Nevertheless, a bigger $L$ entails a larger $\tau_F$ (thus $\alpha_F$) and, in turn, increases the difference between the values of $\frac{R}{R_{NC}}$ for TDD and FDD at high SNR as well.

\subsection{Outage Probability}
We switch our focus to the analysis of the energy shortage and the data outage probability. A set of Monte-Carlo simulations is performed to obtain the numerically computed probabilities. Subsequently, we set the values of $\eta_T$, $\eta_F$, and $\tau_F$ according to Lemma \ref{Lem: TDD eta over all channels} and Lemma \ref{Lem: FDD eta over all channels} and compute the exact value of both metrics by means of the analytic results in Sec.~\ref{Sec: Outage probability}. At this stage, we only consider the case $L=3$ owing to space economy. In the previous subsection, we verified that the impact of change in the number of antennas on the accuracy of the analytic results on the downlink rate is rather small. Accordingly, a robustness of the accuracy of our results to a change in the number of antennas could be conjectured. For the sake of clarity we let $p^{E,out}$ and $p^{D,out}$ be the energy shortage probability and the data outage probability when no energy shortage occurs, respectively. Furthermore, we let $R_{NC}=R_{T}=R_{F}=6$ (bit/s/Hz)\footnote{Referring to Sec.~\ref{sec:data outafe for non-CSI}, \ref{sec:data outafe for TDD}, and \ref{sec:data outafe for FDD}, we note that $R_{NC}$, $R_{T}$, and $R_{F}$ are specified values.} be the target rate for the considered system. Finally, we depict $p^{E,out}$ for SNR$\in[0,30]$~dB in Fig.~\ref{fig:non_energy_outage_L3}, Fig.~\ref{fig:TDD_energy_outage_L3}, and Fig.~\ref{fig:FDD_energy_outage_L3} and $p^{D,out}$ for SNR$\in[0,15]$~dB in Fig.~\ref{fig:all_data_outage_L3}. As shown in these figures, the numerical results perfectly match the analytic results derived in Sec.~\ref{Sec: Outage probability} for all the three schemes. This perfect match verifies the correctness of our derivations.
\begin{figure}[!h]
	\centering
	\includegraphics[width=\columnwidth]{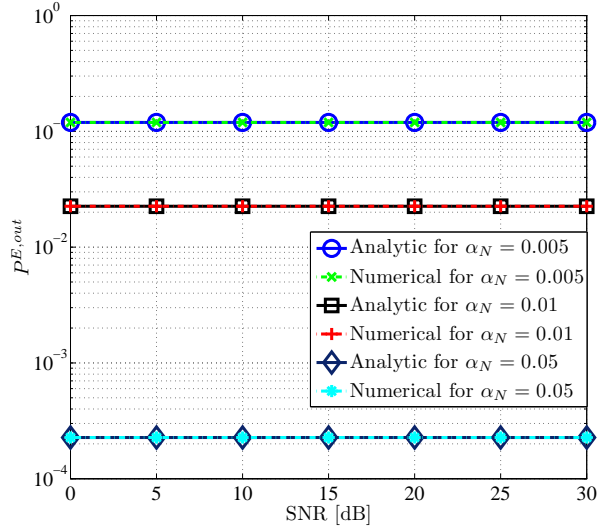}
	\caption{Energy shortage probability, non-CSI and $L=3$ antennas.}	
	\label{fig:non_energy_outage_L3}
\end{figure}
\begin{figure}[!h]
	\centering
	\includegraphics[width=\columnwidth]{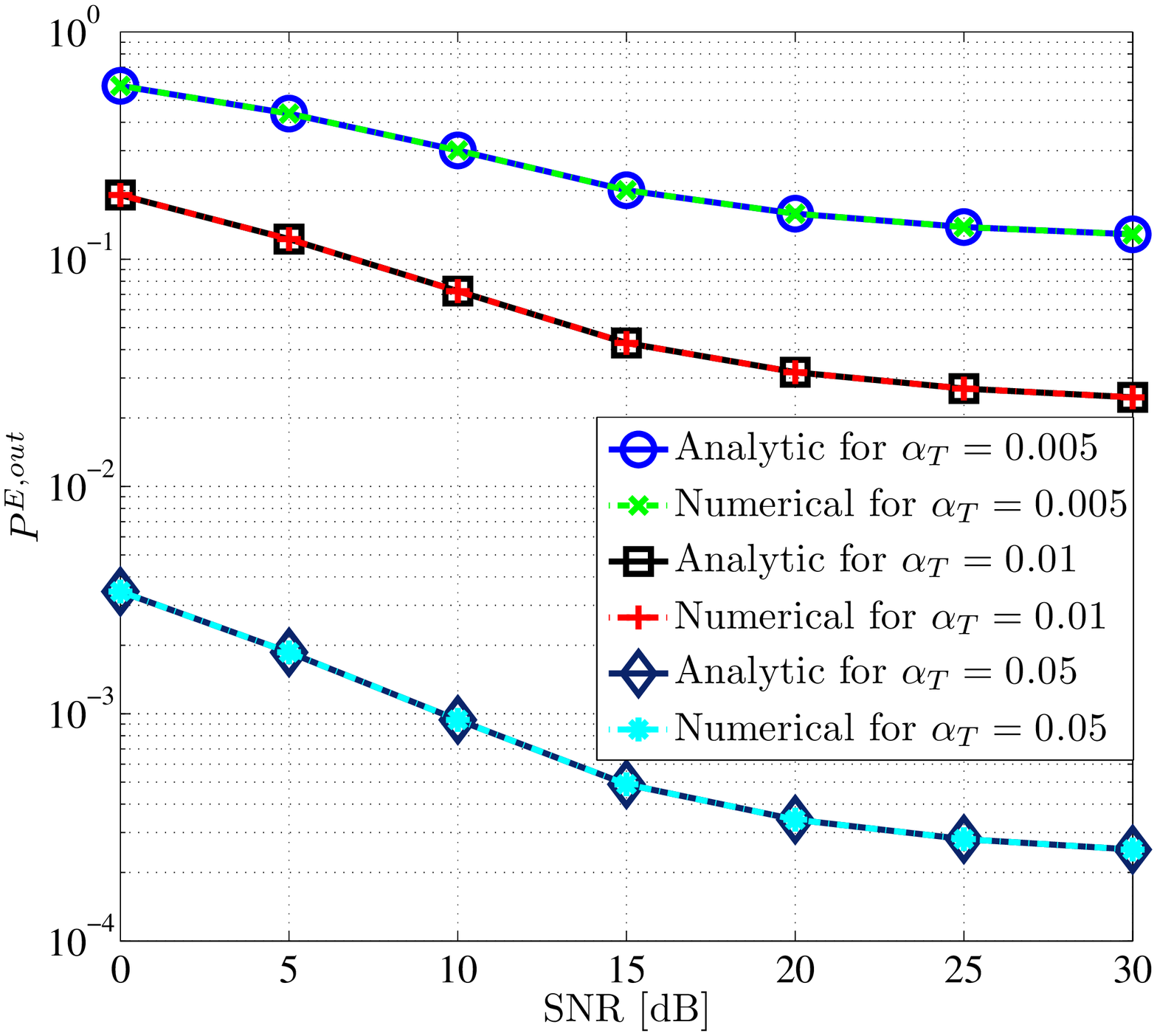}
	\caption{Energy shortage probability, TDD and $L=3$ antennas.}	
	\label{fig:TDD_energy_outage_L3}
\end{figure}
\begin{figure}[!h]
	\centering
	\includegraphics[width=\columnwidth]{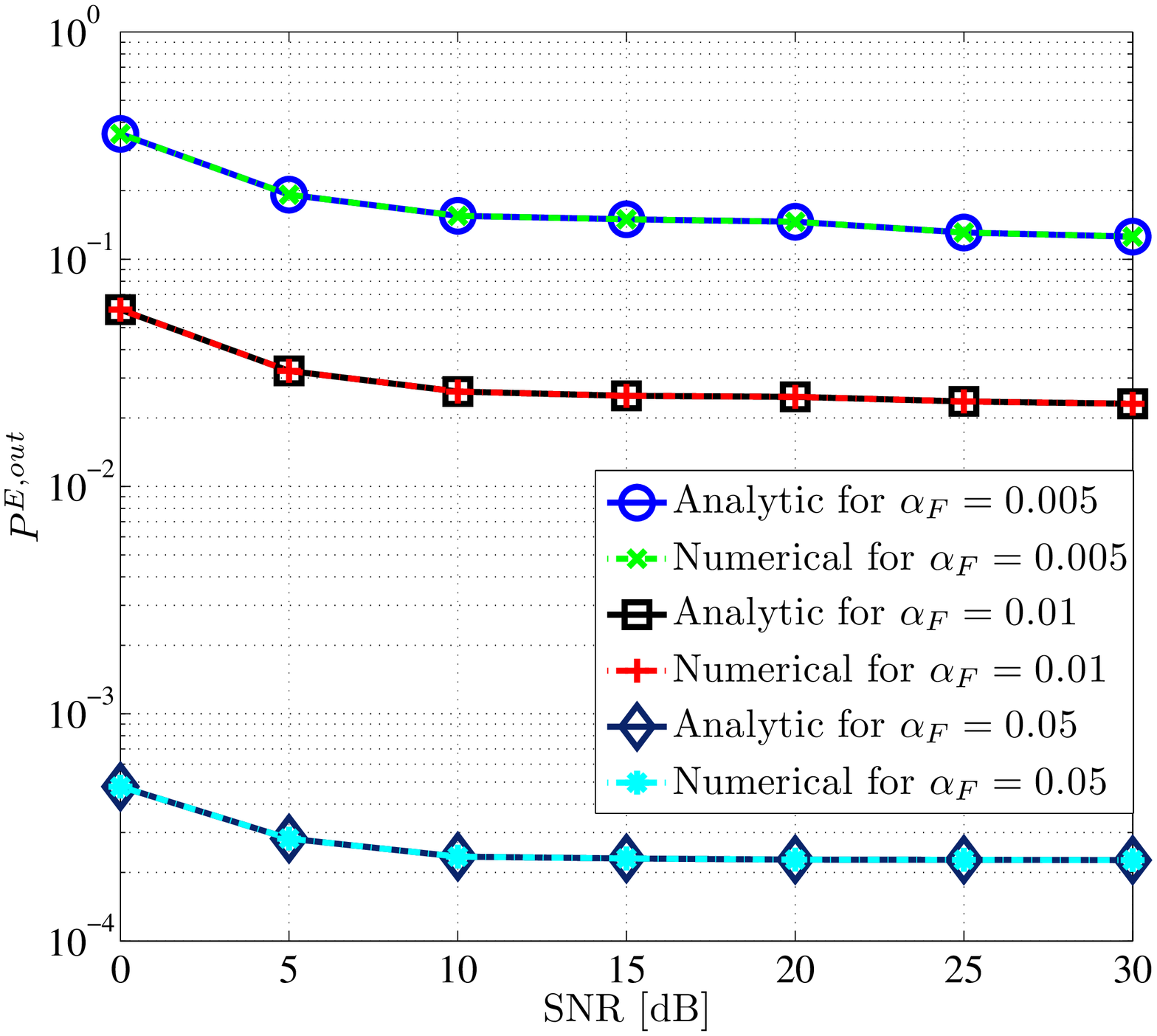}
	\caption{Energy shortage probability, FDD and $L=3$ antennas.}	
	\label{fig:FDD_energy_outage_L3}
\end{figure}
\begin{figure}[!h]
	\centering
	\includegraphics[width=\columnwidth]{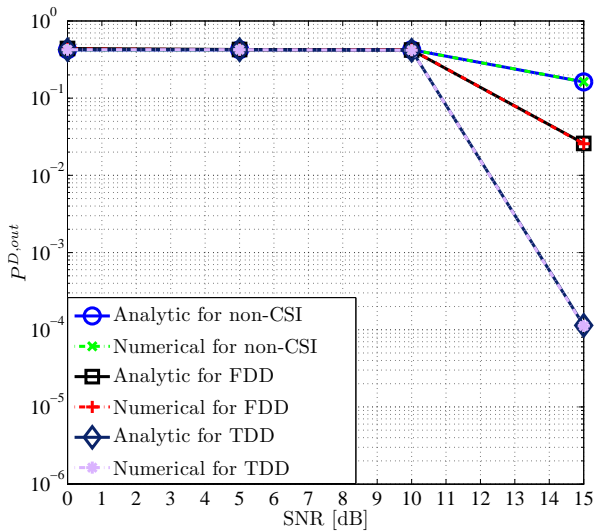}
	\caption{Data outage probability when no energy shortage occurs, $L=3$ antennas.}	
	\label{fig:all_data_outage_L3}
\end{figure}

We start by noting that the energy shortage probability strongly depends on the considered parameters. Thus, a comparison between schemes could have limited interested w.r.t. a comparison between the results obtained for each scheme, as the duration of the energy transfer phase varies. Accordingly, we restrain our focus to the latter aspect.  As expected, the energy shortage probability is independent of the SNR, regardless of the value of $\alpha_N$. However, for both the TDD and the FDD scheme, the energy shortage probability decreases with the SNR, regardless of the value of $\alpha_T$ and $\alpha_F$. In these cases, a larger SNR reduces the optimal time for both devices to perform the operations intrinsic to the CSI acquisition and achieve accurate channel estimations. In other words, the channel estimation accuracy increases with the SNR value, thus the CSI acquisition requires less time. Therefore, the energy consumption at the UT is lower when the SNR is large. Now, if the duration of the energy transfer phase is doubled or tenfold, a reduction of the energy shortage probability from almost one to three orders of magnitude is observed, depending on the considered scheme. In practice, if the coherence time is long enough, even a rather small increase of the duration of the energy transfer phase can positively impact the energy shortage probability.

We now switch our focus to the data outage probability illustrated in Fig.~\ref{fig:all_data_outage_L3}. We start by noting that, to compute the numerical data outage probability in this case, the duration of the energy transfer phase for the three considered schemes, i.e., $\alpha_N$, $\alpha_T$, and $\alpha_F$, is chosen at each iteration of the simulations such that the harvested energy at the UT is sufficient to perform the receiver operations intrinsic to each scheme.  As a matter of fact, the obtained quantitative results for a study of this kind are not extremely relevant, in fact they clearly depend on the selected target rate. In practice, their qualitative behavior is definitely more interesting. In this regard, the lowest data outage probability is experienced by the considered system in the case of the TDD scheme. This could have been expected after our findings on the downlink rate in the previous section, in which the TDD scheme resulted as the best performer out of the three considered cases. 

\section{Conclusion}\label{Sec: Conclusions}

In this work, we have examined the efficacy of SWIPT in a MISO system consisting of an AP and a single UT. In particular, the latter is not equipped with any local power source, but instead harvests the necessary energy for its operations from the received RF signals. The performance of the considered system has been analyzed under realistic and practically relevant system assumptions. Three practical cases have been considered: a) absence of CSI at the AP, b) imperfect CSI at the AP acquired by means of pilots estimation (TDD), c) imperfect CSI at both the UT and the AP acquired by means of analog CSI feedback in the uplink (FDD). We have compared the considered scenarios by means of three performance metrics of interest, i.e, the ergodic downlink rate, the energy shortage probability, and the data outage probability. Accordingly, we have derived closed-form expressions for each metric, and for the ergodically optimal duration of both the WPT and the channel training/feedback phases, to maximize the downlink rate in all the three scenarios. The accuracy of our derivations has been verified by an extensive numerical analysis. First, it is worth noting that TDD has consistently been the best performer for each considered metric, confirming the potential of this duplexing scheme for the future advancements in modern networks. More specifically, concerning the downlink rate, our findings show that CSI knowledge at the AP is always beneficial for the information transfer in SWIPT systems, despite the resources devoted to the channel estimation/feedback procedures and the presence of estimation errors. In a follow-up of this work, we will study both strategies to maximize the efficiency of the WPT, in the case of the availability of CSI knowledge at the AP prior to the WPT phase (or part of it), and their impact on the energy shortage probability. Additional subject of future investigation will be the extension of the considered set-up to a multi-user scenario.

\appendices
\section{Proof of Lemma \ref{Lem: TDD eta over all channels}}
\label{Apx: TDD eta over all channels}
To evaluate \eqref{Eq: SecIII-TDD-rate optimal problem},  we use the law of iterated expectations, i.e., 
$\mathbb{E}_{\mathbf{h},\mathbf{\bar{w}}}\left[\cdot\right]=\mathbb{E}_{\mathbf{h}}\left[\mathbb{E}_{\mathbf{\bar{w}}}\left[\cdot|\mathbf{h}\right]\right]$. Furthermore, we neglect $LP_D$ in \eqref{Eq: SecIII-TDD-rate optimal problem} since, in practice, $P_D\ll P$ generally \cite{rpt:3gpp.36.814}. We proceed by first computing the following expression: 
\begin{align}
\left(1-\eta_T-\frac{\eta_T LP_E}{\beta P\lVert\mathbf{h}\rVert^2}\right)\mathbb{E}_{\mathbf{\bar{w}}}\left[\log_2\left(1+\frac{P\lvert\mathbf{h}^{\dag}\mathbf{\hat{h}}\rvert^2}{N_0\lVert\mathbf{\hat{h}}\rVert^2}\right)\right],\label{Eq: App1-Optimize rate}
\end{align}
for a given channel realization $\mathbf{h}$.
In order to compute \eqref{Eq: App1-Optimize rate}, the following straightforward results can be derived:
\begin{align}
\lvert\mathbf{h}^\dag\hat{\mathbf{h}}\rvert^2=\frac{N_0\lVert\mathbf{h}\rVert^2}{2\eta_T T_CP_E}\Psi_1 \text{ and }
\lVert\mathbf{\hat{h}}\rVert^2=\frac{N_0}{2\eta_T T_CP_E}\Psi_2 \label{eqn:papa101},
\end{align}
where $\Psi_1\sim\chi_2^{'2}\left(\frac{2\eta_T T_CP_E\lVert\mathbf{h}\rVert^2}{N_0}\right)$ and $\Psi_2\sim\chi_{2L}^{'2}\left(\frac{2\eta_T T_CP_E\lVert\mathbf{h}\rVert^2}{N_0}\right)$. We break up the subsequent analysis into two cases, namely, the high SNR and low SNR cases.

First, we consider the analysis at high SNR. In this case, applying the approximation, $\log_2(1+SNR)\approx\log_2SNR$ when $SNR\gg 0$, and \eqref{eqn:papa101} to \eqref{Eq: App1-Optimize rate}, we can derive
\begin{align}
\hspace{-1.3em}\mathbb{E}_{\mathbf{\bar{w}}}\left[\log_2\left(1+\frac{P\lvert\mathbf{h}^{\dag}\mathbf{\hat{h}}\rvert^2}{N_0\lVert\mathbf{\hat{h}}\rVert^2}\right)\right]\approx\mathbb{E}_{\mathbf{\bar{w}}}\left[\log_2\left(\frac{P\lVert\mathbf{h}\rVert^2\Psi_1}{N_0\Psi_2}\right)\right].\label{Eq: App1-High SNR rate approximation-2}
\end{align}
Subsequently, using the Taylor series expansion of $\log_2\Psi_1$ and $\log_2\Psi_2$ 
at their respectively mean values (i.e. $2+\lambda$ and $2L+\lambda$ respectively, where $\lambda=\frac{2\eta_T T_CP_E\lVert\mathbf{h}\rVert^2}{N_0}$), we have
\begin{align}
&\eqref{Eq: App1-High SNR rate approximation-2}= \log_2\frac{P\lVert\mathbf{h}\rVert^2}{N_0}+\log_2e \notag
\\ &\times\mathbb{E}_{\Psi_1,\Psi_2}\Bigg[\ln\left(2+\lambda\right)+\frac{\Psi_1-2-\lambda}{2+\lambda}-\frac{\left(\Psi_1-2-\lambda\right)^2}{2\left(2+\lambda\right)^2}+\cdots\notag
\\&-\ln\left(2L+\lambda\right)-\frac{\Psi_2-2L-\lambda}{2L+\lambda}+\frac{\left(\Psi_2-2L-\lambda\right)^2}{2\left(2L+\lambda\right)^2}-\cdots\Bigg]\label{Eq: TDD eta high SNR Taylor series-1}
\\&=\log_2\frac{P\lVert\mathbf{h}\rVert^2}{N_0}+\log_2e\Bigg(\ln\left(2+\lambda\right)-\frac{\left(2+2\lambda\right)}{\left(2+\lambda\right)^2}+\cdots\notag
\\&-\ln\left(2L+\lambda\right)+\frac{\left(2L+2\lambda\right)}{\left(2L+\lambda\right)^2}-\cdots\Bigg)\label{Eq: TDD eta high SNR Taylor series-1-1}
\\&\stackrel{(a)}{\approx}\log_2\frac{P\lVert\mathbf{h}\rVert^2\lambda}{N_0\left(2L+\lambda\right)}
=\log_2\frac{P\lVert\mathbf{h}\rVert^2}{N_0}-\log_2\left(1+\frac{2L}{\lambda}\right),\label{Eq: Apx2-Approximated rate at high SNR-2}
\end{align}
where $(a)$ in \eqref{Eq: Apx2-Approximated rate at high SNR-2} is derived by noting that $\lambda$ is large (at high SNR), and hence we can neglect the higher order fractional terms in \eqref{Eq: TDD eta high SNR Taylor series-1-1}. 
Moreover,  $2$ in the $\ln(2+\lambda)$ term is neglected owing to $\lambda \gg 2$. 
Further, we use \eqref{Eq: Apx2-Approximated rate at high SNR-2} in \eqref{Eq: App1-Optimize rate} and take the expectation over $\mathbf{h}$.
Moreover, since $L/\lambda$ is small at high SNR, we use the approximation, $\log_2\left(1+x \right)\approx x\log_2e$ when $x\approx0$, in the derivation. By some straightforward computations, we can rewrite \eqref{Eq: SecIII-TDD-rate optimal problem} as
\begin{align}
\eta_T^{\star}\approx\operatorname{argmax}_{\eta_T}
B_2-\eta_TB_1-\frac{N_0L\log_2e}{\eta_TT_CP_E\left(L-1\right)},\label{Eq: Apx2-Approximated rate at high SNR-4}
\end{align}
where
$B_1=\mathbb{E}_{\mathbf{h}}\left[\left(1+\frac{LP_E}{\beta P\lVert\mathbf{h}\rVert^2}
\right)\log_2\left(\frac{P\lVert\mathbf{h}\rVert^2}
{N_0}\right)\right]$, and $B_2$ is a constant value. 
In order to derive $\eta_T^{\star}$,
we differentiate \eqref{Eq: Apx2-Approximated rate at high SNR-4} with respect $\eta_T$ and set it equal to $0$. By some straightforward computations, we can derive  the result \eqref{Eq: TDD_opt_high_SNR}.

We now move to the analysis at low SNR. Before deriving, we first note that $\mathbb{E}[X]=\mathbb{E}[\exp(\ln X)]$. Subsequently,
using the approximation, $\log_2\left(1+SNR \right)\approx SNR\log_2e$ when $SNR\approx0$, and Jensen's inequality, i.e., $\mathbb{E}[\exp(\ln X)]\geq\exp(\mathbb{E}[\ln X])$, we have
\begin{align}
\eqref{Eq: App1-Optimize rate}\geq\left(1-\eta_T-\frac{\eta_T LP_E}{\beta P\lVert\mathbf{h}\rVert^2}\right)\frac{P\log_2e}{N_0}&\notag
\\\times \exp\left(\mathbb{E}_{\mathbf{\bar{w}}}\left[\ln\left(\frac{\lvert\mathbf{h}^{\dag}\mathbf{\hat{h}}\rvert^2}{\lVert\mathbf{\hat{h}}\rVert^2}\right)\right]\right)&.\label{Eq: TDD eta low SNR Jensen's inequality}
\end{align}
Once again, we apply Taylor series expansion to $\mathbb{E}_{\mathbf{\bar{w}}}\left[\ln\left(\frac{\lvert\mathbf{h}^{\dag}\mathbf{\hat{h}}\rvert^2}{\lVert\mathbf{\hat{h}}\rVert^2}\right)\right]$, i.e., steps \eqref{Eq: App1-High SNR rate approximation-2}, \eqref{Eq: TDD eta high SNR Taylor series-1}, \eqref{Eq: TDD eta high SNR Taylor series-1-1}, and \eqref{Eq: Apx2-Approximated rate at high SNR-2}. In this case, since $\lambda$ is small at low SNR, we approximate the higher order terms, such as $\frac{2+2\lambda}{(2+\lambda)^2}$ in \eqref{Eq: TDD eta high SNR Taylor series-1-1}, by
a constant value $\kappa_1$. Using this, we can rewrite \eqref{Eq: TDD eta low SNR Jensen's inequality} as
\begin{align}  
\frac{\kappa_1\left(1-\eta_T-\frac{\eta_T LP_E}{\beta P\lVert\mathbf{h}\rVert^2}\right)P\lVert\mathbf{h}\rVert^2\log_2e}{N_0}\times\frac{2+\lambda}{2L+\lambda}.\label{Eq: TDD eta low SNR constant fraction}
\end{align}
Finally, we take the expectation over $\mathbf{h}$. By using Jensen's inequality (similar approaches in \eqref{Eq: TDD eta low SNR Jensen's inequality}) and applying Taylor series expansion 
for the logarithm term
at the mean value $\lVert\mathbf{h}\rVert^2$ (similar approaches in \eqref{Eq: TDD eta high SNR Taylor series-1}, \eqref{Eq: TDD eta high SNR Taylor series-1-1}, and \eqref{Eq: Apx2-Approximated rate at high SNR-2}), we rewrite \eqref{Eq: SecIII-TDD-rate optimal problem} as
\begin{align}
\eta_T^{\star}\approx\operatorname{argmax}_{\eta_T}
\frac{\kappa_2\left(1+\frac{\eta_TT_CP_EL}{N_0}\right)\left(L-\eta_TL-\frac{\eta_T LP_E}{\beta P}\right)}{1+\frac{\eta_TT_CP_E}{N_0}}\notag
%\label{Eq: Apx2-Approximated rate at low SNR-2}
\end{align}
where $\kappa_2$ is a constant. 
In order to derive $\eta_T^{\star}$,
we differentiate the above formula 
%\eqref{Eq: Apx2-Approximated rate at low SNR-2} 
with respect $\eta_T$ and set it equal to $0$. By some straightforward computations, we can derive the result \eqref{Eq: TDD_opt_low_SNR}.

%%%%%%%%%%%%%%%%%%%%%%%%%%%%%%%%%%%%%%%%%%%%%%%%%%%%%%%%%%%%%%%%%%%%%%%%%%%%%%%%%%%%%%%%%%%%%%%%%%%%%%%%%%

%%%%%%%%%%%%%%%%%%%%%%%%%Commenting the previous Appendix D - FDD eta over all channels%%%%%%%%%%%%%%%%%%%%%%%%%%%%%%%%%%%%%%
\section{Proof of Lemma \ref{Lem: FDD eta over all channels}}\label{Apx: FDD eta over all channels}
To evaluate \eqref{Eq: SecIII-FDD-rate optimal problem}, we use the similar approaches as in \eqref{Eq: App1-Optimize rate}. Hence, we first compute
\begin{align}
\mathbb{E}_{\mathbf{\hat{w}}}\Bigg[\left(1-\tau_F-\frac{\tau_FP_F\lVert\mathbf{\hat{h}}_{UT}\rVert^2}{\beta P\lVert\mathbf{h}\rVert^2}-\eta_F\right)&\notag
\\\times\log_2\left(1+\frac{P\lvert\mathbf{h}^{\dag}\mathbf{\hat{h}}_{AP}\rvert^2}{N_0\lVert\mathbf{\hat{h}}_{AP}\rVert^2}\right)&\Bigg].\label{Eq: App2-Optimize rate}
\end{align}
To compute \eqref{Eq: App2-Optimize rate},
the following results can be derived (details omitted for lack of space):
\begin{align}
\lvert\mathbf{h}^{\dag}\mathbf{\hat{h}}_{AP}\rvert^2&=\frac{N_0L\lVert\mathbf{h}\rVert^2}{2T_C}\left(\frac{1}{\eta_F P}+\frac{1}{\tau_F P_F}\right)\Phi_1,\label{Eq: App2-FDD distribution-1}
\\\lVert\mathbf{\hat{h}}_{AP}\rVert^2&=\frac{N_0L}{2T_C}\left(\frac{1}{\eta_F P}+\frac{1}{\tau_F P_F}\right)\Phi_2,\label{Eq: App2-FDD distribution-2}
\\\lVert\mathbf{\hat{h}}_{UT}\rVert^2&=\frac{N_0L}{2\eta_FT_C P}\Phi_3,\label{Eq: App2-FDD distribution-3}
\end{align}
where $\Phi_1\sim\chi_{2}^{'2}\left(\frac{2T_C\lVert\mathbf{h}\rVert^2}{N_0L\left(\frac{1}{\eta_F P}+\frac{1}{\tau_F P_F}\right)}\right)$, $\Phi_2\sim\chi_{2L}^{'2}\left(\frac{2T_C\lVert\mathbf{h}\rVert^2}{N_0L\left(\frac{1}{\eta_F P}+\frac{1}{\tau_F P_F}\right)}\right)$, and $\Phi_3\sim\chi^{'2}_{2L}\left(\frac{2\eta_FT_C P\lVert\mathbf{h}\rVert^2}{N_0L}\right)$. 

Before proceeding, we note that the pre-log term and the term inside the logarithm in \eqref{Eq: App2-Optimize rate} are correlated, due to the presence of $\mathbf{\hat{w}}_{UT}$ in both terms. However, at high SNR, the variance of $\mathbf{\hat{w}}_{UT}$ will be small. Thus, we assume that the pre-log term and term inside the logarithm are approximately independent (and hence we can take the expectations of these two terms in \eqref{Eq: App2-Optimize rate} separately). For the term inside the logarithm, using \eqref{Eq: App2-FDD distribution-1}, \eqref{Eq: App2-FDD distribution-2}, and \eqref{Eq: App2-FDD distribution-3} and by Taylor series expansion (similar approaches in \eqref{Eq: TDD eta high SNR Taylor series-1}, \eqref{Eq: TDD eta high SNR Taylor series-1-1}, and \eqref{Eq: Apx2-Approximated rate at high SNR-2}), we obtain
\begin{align}
\mathbb{E}_{\mathbf{\hat{w}}}\left[\log_2\left(1+\frac{P\lvert\mathbf{h}^{\dag}\mathbf{\hat{h}}_{AP}\rvert^2}{N_0\lVert\mathbf{\hat{h}}_{AP}\rVert^2}\right)\right]
\approx
\log_2\frac{P\lVert\mathbf{h}\rVert^2\lambda_1}{N_0\left(2L+\lambda_1\right)},\label{Eq: FDD eta high SNR neglect fraction-2}
\end{align}
where $\lambda_1=\frac{2T_C\lVert\mathbf{h}\rVert^2}{N_0L\left(\frac{1}{\eta_F P}+\frac{1}{\tau_F P_F}\right)}$.
For the pre-log term, we have
\begin{align}
&\mathbb{E}_{\mathbf{\hat{w}}}\left[1-\tau_F-\frac{\tau_F P_F\lVert\mathbf{\hat{h}}_{UT}\rVert^2}{\beta P\lVert\mathbf{h}\rVert^2}-\eta_F\right]\approx1-\tau_F-\frac{\tau_F P_F}{\beta P}-\eta_F.
\label{Eq: FDD eta high SNR neglect fraction-1}
\end{align}
The approximation in \eqref{Eq: FDD eta high SNR neglect fraction-1} is derived using $\mathbb{E}_{\mathbf{\hat{w}}}[\lVert\mathbf{\hat{h}}_{UT}\rVert^2]
=\lVert\mathbf{h}\rVert^2+\frac{N_0L^2}{\eta_FT_C P}\approx\lVert\mathbf{h}\rVert^2$
(since at high SNR, $\frac{N_0L^2}{\eta_F PT_C}$ is small enough to be neglected).
Finally, we substitute \eqref{Eq: FDD eta high SNR neglect fraction-2} and \eqref{Eq: FDD eta high SNR neglect fraction-1} into \eqref{Eq: App2-Optimize rate} and take the expectation over $\mathbf{h}$. 
Using the approximation $\log_2\left(1+x \right)\approx x\log_2e$ when $x\approx0$ and by some straightforward computations, we can rewrite \eqref{Eq: SecIII-FDD-rate optimal problem} as
\begin{align}
\left(\eta_F^{\star},\tau_F^{\star}\right)\approx\operatorname{argmax}_{\eta_F,\tau_F}\left(1-\tau_F-\frac{\tau_FP_F}{\beta P}-\eta_F\right)&\notag
\\\times\left(B_5-\frac{N_0L^2\log_2e}{T_C\left(L-1\right)}\left(\frac{1}{\eta_F P}+\frac{1}{\tau_F P_F}\right)\right)&,\label{Eq: Apx4-Objective function at high SNR-3}
\end{align}
where $B_5=\mathbb{E}_\mathbf{h}\left[\log_2\left(\frac{P\lVert\mathbf{h}\rVert^2}{N_0}\right)\right]$. 
In order to find $\eta_F^{\star}$ and $\tau_F^{\star}$, we differentiate \eqref{Eq: Apx4-Objective function at high SNR-3} with respect $\eta_F$ and $\tau_F$ and equate them to $0$. By some straightforward computations, we can derive  the results \eqref{Eq: FDD_opt_high_SNR-1} and \eqref{Eq: FDD_opt_high_SNR-2}.

We now turn to the analysis at low SNR. Herein, we follow the similar derivations as in Appendix-\ref{Apx: TDD eta over all channels}. First we use the same approaches as in \eqref{Eq: TDD eta low SNR Jensen's inequality}.
Further, we apply \eqref{Eq: App2-FDD distribution-1}, \eqref{Eq: App2-FDD distribution-2}, and \eqref{Eq: App2-FDD distribution-3} and use the Taylor series expansion (similar to the approaches in \eqref{Eq: TDD eta low SNR constant fraction}).
Thus, we can derive
\begin{align}
\eqref{Eq: App2-Optimize rate} &\geq
\left(\left(1-\eta_F-\tau_F\right)\beta P\lVert\mathbf{h}\rVert^2-\frac{\tau_F P_FN_0L}{2\eta_FT_CP}\left(2L+\lambda_2\right)\right)\notag
\\&\hspace{11em}\times\frac{\kappa_3 \log_2e\left(2+\lambda_1\right)}{N_0\beta \left(2L+\lambda_1\right)},
\label{Eq: FDD eta low SNR Taylor series-2}
\end{align}
where $\lambda_2=\frac{2\eta_F T_CP\lVert\mathbf{h}\rVert^2}{N_0L}$, and $\kappa_3$ is a constant value.
Finally, as before we take the expectation over $\mathbf{h}$. By using Jensen's inequality (similar approaches in \eqref{Eq: TDD eta low SNR Jensen's inequality}) and applying Taylor series expansion of the logarithm term at the mean value $\lVert\mathbf{h}\rVert^2$ (similar approaches in \eqref{Eq: TDD eta high SNR Taylor series-1}, \eqref{Eq: TDD eta high SNR Taylor series-1-1}, and \eqref{Eq: Apx2-Approximated rate at high SNR-2}), we can rewrite the expected value of \eqref{Eq: FDD eta low SNR Taylor series-2} over $\mathbf{h}$ as
\begin{align}
&\frac{\kappa_4\left(1-\eta_F-\tau_F-\frac{\tau_F P_F}{\beta P}
-\frac{\tau_F P_FN_0L}{\eta_FT_C\beta P^2}\right)}{\left(L+\frac{T_C}{N_0\left(\frac{1}{\eta_F P}+\frac{1}{\tau_F P_F}\right)}\right)\left(1+\frac{T_C}{N_0\left(\frac{1}{\eta_F P}+\frac{1}{\tau_F P_F}\right)}\right)^{-1}},\label{Eq: Apx4-Low SNR approximation-4}
\end{align}
where $\kappa_4$ is a constant value. Now, we differentiate \eqref{Eq: Apx4-Low SNR approximation-4} with respect $\eta_F$ and $\tau_F$ and equate them to $0$. Using some straightforward computations, we obtain the result \eqref{Eq: FDD_opt_low_SNR-1}.
When $x\approx 0$, $\sqrt{1+x}\approx 1+\frac{x}{2}$. By utilizing this approximation in \eqref{Eq: FDD_opt_low_SNR-1} since $N_0$ is large at low SNR, we can derive $\tau^{\star}_F\approx\frac{\left(\eta^{\star}_F\right)^2T_C\beta P^2}{P_FN_0L}$.
Substituting the latter approximated $\tau^{\star}_F$ into the differentiated equation $\frac{\partial\eqref{Eq: Apx4-Low SNR approximation-4}}{\partial \eta_F}=0$, we can also derive the result \eqref{Eq: FDD_opt_low_SNR-2}.

\section{Proof of Lemma \ref{Lem: TDD data outage probability}}\label{Apx: TDD data outage probability}

We now proceed with the proof. In equation \eqref{Eq: TDD-Data outage}, we have two random variables, i.e., $\left|\frac{\mathbf{h}^{\dag}\mathbf{\hat{h}}}{\lVert\mathbf{\hat{h}}\rVert}\right|^2$ and $\lVert\mathbf{h}\rVert^2$. 
To evaluate \eqref{Eq: TDD-Data outage}, we first express
$\lVert\mathbf{h}\rVert^2$ as the sum of $\left|\frac{\mathbf{h}^{\dag}\mathbf{\hat{h}}}{\lVert\mathbf{\hat{h}}\rVert}\right|^2$ and another independent random variable (see the steps below). This step simplifies the evaluation as shown in the following.
In order to do so, first we project the row vector $\mathbf{h}^{\dag}$ onto an orthonormal set of vectors $\Omega = \left\{\frac{\mathbf{\hat{h}}^{\dag}}{\lVert\mathbf{\hat{h}}\rVert},\mathbf{g}_2^{\dag},\cdots,\mathbf{g}_L^{\dag}\right\}$, 
where $\mathbf{g}_2^{\dag},\cdots,\mathbf{g}_L^{\dag}$ are chosen arbitrarily 
such that the vectors in $\Omega $ span the complex $L$ dimensional space. 
Recall that $\mathbf{h}^{\dag}|_{\mathbf{\hat{h}}}\sim\mathcal{C}\mathcal{N}\left(\frac{1}{1+b_6}\mathbf{\hat{h}}^{\dag},\frac{1}{b_5}\mathbf{I}_L\right)$. 
Since the distribution of a Complex Gaussian random vector (with distribution $\mathcal{C}\mathcal{N}\left(\mathbf{0},\mathbf{I}\right)$) projected onto an orthonormal set remains
unchanged \cite{TseWireless}, we can conclude that
%
%
%\begin{align*} 
$ \frac{\mathbf{h}^{\dag}\mathbf{\hat{h}}}{\lVert\mathbf{\hat{h}}\rVert}\Big|_{\mathbf{\hat{h}}}\sim\mathcal{C}\mathcal{N}\left(\frac{\lVert\mathbf{\hat{h}}\rVert}{\left(1+b_6\right)},b_5^{-1}\right)$ and 
$\mathbf{h}^{\dag}\mathbf{g}_l\big|_{\mathbf{\hat{h}}}\sim\mathcal{C}\mathcal{N}\left(0,b_5^{-1}\right),
\ l = 2,\dots,L.$
%\end{align*}
%
%
Additionally, the $\mathbf{h}^{\dag}\mathbf{g}_l\big|_{\mathbf{\hat{h}}} \ l = 2,\dots,L$  are independent random variables.
Thus we can conclude the following:
\begin{align}
\left|\frac{\mathbf{h}^{\dag}\mathbf{\hat{h}}}{\lVert\mathbf{\hat{h}}\rVert}\right|^2\Bigg|_{\mathbf{\hat{h}}}&=\frac{1}{2b_5}\Theta_1,\label{Eq: App3-chi-square-1}
\\\left(\left|\mathbf{h}^{\dag}\mathbf{g}_2\right|^2+\cdots+\left|\mathbf{h}^{\dag}\mathbf{g}_L\right|^2\right)\Big|_{\mathbf{\hat{h}}}&=\frac{1}{2b_5}\Theta_2,\label{Eq: App3-chi-square-2}
\end{align}
where $\Theta_1\sim\chi^{'2}_2\left(\frac{2}{b_5b_6^2}\lVert\mathbf{\hat{h}}\rVert^2\right)$, and $\Theta_2\sim\chi^{2}_{2L-2}$. Furthermore, $\Theta_1$ and $\Theta_2$ are independent and $\Theta_1+\Theta_2=2b_5\lVert\mathbf{h}\rVert^2$. 
Now, applying the same approach as in \eqref{Eq: App4-Outage probability-1} to \eqref{Eq: TDD-Data outage} and using \eqref{Eq: App3-chi-square-1} and \eqref{Eq: App3-chi-square-2}, we can derive the following:
$\mathcal{P}_{T}^{D, out}=\mathbb{E}_\mathbf{\hat{h}}\left[\Pr\left\{
\Theta_1<2b_5b_3,\Theta_1+\Theta_2\geq 2b_5b_4
\right\}\right].$
Let us focus on the relationship between $b_3$ and $b_4$ and consider two possible cases, i.e., $b_3<b_4$ and $b_3\geq b_4.$ We start from the former. In this case, denoting the PDF of $\Theta_i$  as $f_{\Theta_i}\left(\theta_i \right)$ for $i\in \{1,2\}$,
we have,
\begin{align}
&\mathcal{P}_{T}^{D, out}\notag
\\&=\mathbb{E}_\mathbf{\hat{h}}\left[\int_{\theta_1 = 0}^{2b_5b_3}\frac{\Gamma\left(L-1,b_5 b_4-\frac{\theta_1}{2}\right)I_0\left(\sqrt{\frac{2\lVert\mathbf{\hat{h}}\rVert^2\theta_1}{b_5b_6^2}}\right)}{2\Gamma\left(L-1\right)e^{\left(\frac{\theta_1}{2}+\frac{\lVert\mathbf{\hat{h}}\rVert^2}{b_5b_6^2}\right)}}
\mathrm{d}\theta_1\right].\label{Eq: App3-TDD data outage a<b-1}
\end{align}
Since $\mathbf{\hat{h}}\sim\mathcal{C}\mathcal{N}\left(\mathbf{0},\left(1+b_6\right)\mathbf{I}_L\right)$, it follows that $\lVert\mathbf{\hat{h}}\rVert^2=\frac{\left(1+b_6\right)}{2}\Theta_3$, where $\Theta_3\sim\chi^2_{2L}$. Substituting the PDF of $\lVert\mathbf{\hat{h}}\rVert^2$ into \eqref{Eq: App3-TDD data outage a<b-1}, we obtain our result \eqref{Eq: TDD data outage probability closed form 1}.
In the second case $b_3\geq b_4$,
following the same approaches as in \eqref{Eq: App3-TDD data outage a<b-1}, we obtain our result \eqref{Eq: TDD data outage probability closed form 2} and conclude the proof.

\section{Proof of Lemma \ref{Lem: FDD data outage probability}}\label{Apx: FDD data outage probability}

Adopting a similar approach to what has been adopted in \eqref{Eq: App4-Outage probability-1} and denoting the conditional distribution of $\mathbf{\hat{h}}_{UT}$ given $\mathbf{\hat{h}}_{AP}$ as 
$f\left(\mathbf{\hat{h}}_{UT}\big|\mathbf{\hat{h}}_{AP}\right)$, the analytic expression for the outage probability, i.e., $\mathcal{P}_{F}^{D, out}$ in \eqref{Eq: FDD-Data outage-1}, can be written as
\begin{align}
&\mathbb{E}_{\mathbf{\hat{h}}_{AP}}\bigg[\int\Pr\Bigg\{
\left|\frac{\mathbf{h}^{\dag}\mathbf{\hat{h}}_{AP}}{\lVert\mathbf{\hat{h}}_{AP}\rVert}\right|^2<b_7,\lVert\mathbf{h}\rVert^2\geq b_8\notag
\\&\quad\quad+b_9\lVert\mathbf{\hat{h}}_{UT}\rVert^2\Big|\mathbf{\hat{h}}_{AP},\mathbf{\hat{h}}_{UT}\Big\}
f\left(\mathbf{\hat{h}}_{UT}\big|\mathbf{\hat{h}}_{AP}\right)\mathrm{d}\mathbf{\hat{h}}_{UT}\Big].\label{Eq: Proof of data outage in FDD-1}
\end{align}
To compute \eqref{Eq: Proof of data outage in FDD-1}, we use the projection approach that we have used in Appendix-\ref{Apx: TDD data outage probability}. First, we project the row vector $\mathbf{h}^{\dag}$ onto an orthonormal set of vector $\left\{\frac{\mathbf{\hat{h}}^{\dag}_{AP}}{\lVert\mathbf{\hat{h}}_{AP}\rVert},\mathbf{g}^{\dag}_2,\cdots,\mathbf{g}_L^{\dag}\right\}$, (such that
these vectors span the $L$ dimensional complex space). 
Since $\mathbf{h}|_{\mathbf{\hat{h}}_{AP},\mathbf{\hat{h}}_{UT}}\sim\mathcal{C}\mathcal{N}\left(\frac{1}{1+\sigma_3}\mathbf{\hat{h}}_{UT},\frac{1}{\sigma_2}\mathbf{I}_L\right)$, following the same approach as in Appendix-\ref{Apx: TDD data outage probability}, we obtain 
\begin{align}
\left|\frac{\mathbf{h}^{\dag}\mathbf{\hat{h}}_{AP}}{\lVert\mathbf{\hat{h}}_{AP}\rVert}\right|^2\Bigg|_{\mathbf{\hat{h}}_{AP},\mathbf{\hat{h}}_{UT}}&=\frac{1}{2\sigma_2}\Theta_5,\label{Eq: App-5-chi square-1}
\\\left(\left|\mathbf{h}^{\dag}\mathbf{g}_2\right|^2+\cdots\left|\mathbf{h}^{\dag}\mathbf{g}_L\right|^2\right)\Big|_{\mathbf{\hat{h}}_{AP},\mathbf{\hat{h}}_{UT}}&=\frac{1}{2\sigma_2}\Theta_6,\label{Eq: App-5-chi square-2}
\end{align}
where $\Theta_5\sim\chi^{'2}_2\left(\frac{2}{\sigma_2\sigma_3^2}\left|\frac{\mathbf{\hat{h}}_{UT}^{\dag}\mathbf{\hat{h}}_{AP}}{\lVert\mathbf{\hat{h}}_{AP}\rVert}\right|^2\right)$ and $\Theta_6\sim\chi^{'2}_{2L-2}\left(\frac{2}{\sigma_2\sigma_3^2}\left(\lVert\mathbf{\hat{h}}_{UT}\rVert^2-\left|\frac{\mathbf{\hat{h}}_{UT}^{\dag}\mathbf{\hat{h}}_{AP}}{\lVert\mathbf{\hat{h}}_{AP}\rVert}\right|^2\right)\right)$. Moreover, $\Theta_5$ and $\Theta_6$ are independent and $\Theta_5+\Theta_6=2\sigma_2\lVert\mathbf{h}\rVert^2$. Using \eqref{Eq: App-5-chi square-1} and \eqref{Eq: App-5-chi square-2}, we rewrite \eqref{Eq: Proof of data outage in FDD-1} as
\begin{align}
&\mathbb{E}_{\mathbf{\hat{h}}_{AP}}\bigg[\int\Pr\Big\{
\Theta_5+\Theta_6\geq 2\sigma_2 \left(b_8+b_9\lVert\mathbf{\hat{h}}_{UT}\rVert^2\right),\notag
\\&\hspace{7em}\Theta_5<2\sigma_2 b_7
\Big\}
f\left(\mathbf{\hat{h}}_{UT}\big|\mathbf{\hat{h}}_{AP}\right)\mathrm{d}\mathbf{\hat{h}}_{UT}\Big] \label{Eq: Proof of data outage in FDD-2}.
\end{align}
To compute the double integration over $\theta_5$ and $\theta_6$ in \eqref{Eq: Proof of data outage in FDD-2}, we have two two cases, i.e., $b_7<b_8+b_9\lVert\mathbf{\hat{h}}_{UT}\rVert^2$ and $b_7\geq b_8+b_9\lVert\mathbf{\hat{h}}_{UT}\rVert^2$. In the first case, the probability term in \eqref{Eq: Proof of data outage in FDD-2} can be computed as
\begin{align}
&\int_{\theta_5=0}^{2\sigma_2 b_7}\frac{1}{2}Q_{L-1}\left(\sqrt{\frac{2}{\sigma_2\sigma_3^2}\left(\lVert\mathbf{\hat{h}}_{UT}\rVert^2-\left|\frac{\mathbf{\hat{h}}_{UT}^{\dag}\mathbf{\hat{h}}_{AP}}{\lVert\mathbf{\hat{h}}_{AP}\rVert}\right|^2\right)},\right.\notag
\\&\hspace{0.25em}\left.\sqrt{2\sigma_2\left(b_8+b_9\lVert\mathbf{\hat{h}}_{UT}\rVert^2\right)-\theta_5}\right)
I_0\left(\sqrt{\frac{2\theta_5}{\sigma_2\sigma_3^2}\left|\frac{\mathbf{\hat{h}}_{UT}^{\dag}\mathbf{\hat{h}}_{AP}}{\lVert\mathbf{\hat{h}}_{AP}\rVert}\right|^2}\right)\notag
\\&\quad\times e^{-\left(\frac{\theta_5}{2}+\frac{1}{\sigma_2\sigma_3^2}\left|\frac{\mathbf{\hat{h}}_{UT}^{\dag}\mathbf{\hat{h}}_{AP}}{\lVert\mathbf{\hat{h}}_{AP}\rVert}\right|^2\right)}\mathrm{d}\theta_5.\label{Eq: Proof of data outage in FDD-3}
\end{align}
Since $\mathbf{\hat{h}}_{UT}|_{\mathbf{\hat{h}}_{AP}}\sim\mathcal{C}\mathcal{N}\left(\frac{\sigma_5}{\sigma_4}\mathbf{\hat{h}}_{AP},\sigma_5\mathbf{I}_L\right)$, using the projection approach once again, we can derive 
\begin{align}
\left|\frac{\mathbf{\hat{h}}_{UT}^{\dag}\mathbf{\hat{h}}_{AP}}{\lVert\mathbf{h}_{AP}\rVert}\right|^2\Bigg|_{\mathbf{\hat{h}}_{AP}}&=\frac{\sigma_5}{2}\Theta_7,\label{Eq: App-5-chi square-3}
\\\left(\lVert\mathbf{\hat{h}}_{UT}\rVert^2-\left|\frac{\mathbf{\hat{h}}_{UT}^{\dag}\mathbf{\hat{h}}_{AP}}{\lVert\mathbf{h}_{AP}\rVert}\right|^2\right)\Bigg|_{\mathbf{\hat{h}}_{AP}}&=\frac{\sigma_5}{2}\Theta_8,\label{Eq: App-5-chi square-4}
\end{align}
where $\Theta_7\sim\chi^{'2}_2\left(\frac{2\sigma_5}{\sigma_4^2}\lVert\mathbf{\hat{h}}_{AP}\rVert^2\right)$ and $\Theta_8\sim\chi^2_{2L-2}$. Furthermore, $\Theta_7$ and $\Theta_8$ are independent and $\Theta_7+\Theta_8=\frac{2}{\sigma_5}\lVert\mathbf{\hat{h}}_{UT}\rVert^2$. 
Lastly, we note that since $b_7<b_8+b_9\lVert\mathbf{\hat{h}}_{UT}\rVert^2$, we have 
$b_7<(b_8+b_9) \frac{(\Theta_7+\Theta_8)\sigma}{2}$ 
and hence $\Theta_7+\Theta_8>\frac{2\left(b_7-b_8\right)}{b_9\sigma_5}.$
Now, applying \eqref{Eq: Proof of data outage in FDD-3}, \eqref{Eq: App-5-chi square-3}, and \eqref{Eq: App-5-chi square-4} to \eqref{Eq: Proof of data outage in FDD-2}, the integral of \eqref{Eq: Proof of data outage in FDD-2} can be computed as
\begin{align}
&\hspace{-0.5em}\int_{\theta_7+\theta_8>\frac{2\left(b_7-b_8\right)}{b_9\sigma_5}}\int_{\theta_5 = 0}^{2\sigma_2 b_7}
I_0\left(\sqrt{\frac{2\sigma_5\theta_7\lVert\mathbf{\hat{h}}_{AP}\rVert^2}{\sigma_4^2}}\right)\frac{\theta_8^{L-2}}{2^{L+1}}\notag
\\&\hspace{-0.5em}\times Q_{L-1}\left(\sqrt{\frac{\theta_8\sigma_5}{\sigma_2\sigma_3^2}},
\sqrt{2\sigma_2\left(b_8+\frac{b_9\left(\theta_7+\theta_8\right)\sigma_5}{2}\right)-\theta_5}\right)\notag
\\&\hspace{-0.5em}\times\frac{I_0\left(\sqrt{\frac{\theta_5\theta_7\sigma_5}{\sigma_2\sigma_3^2}}\right)}{\Gamma\left(L-1\right)e^{\left(\frac{\theta_5+\theta_7+\theta_8}{2}+\frac{\theta_7\sigma_5}{2\sigma_2\sigma_3^2}+\frac{\sigma_5\lVert\mathbf{\hat{h}}_{AP}\rVert^2}{\sigma_4^2}\right)}}\mathrm{d}\theta_5\mathrm{d}\theta_7\mathrm{d}\theta_8.\label{Eq: Proof of data outage in FDD-4}
\end{align}
We note that, since $\mathbf{\hat{h}}_{AP}\sim\mathcal{C}\mathcal{N}\left(\mathbf{0},\left(1+\sigma_3+\sigma_4\right)\mathbf{I}_L\right)$, we have that $\lVert\mathbf{\hat{h}}_{AP}\rVert^2=\frac{1+\sigma_3+\sigma_4}{2}\Theta_9$, where $\Theta_9$ is $\chi^2_{2L}$. Using this fact and \eqref{Eq: Proof of data outage in FDD-4} in \eqref{Eq: Proof of data outage in FDD-2}, we obtain \eqref{Eq: FDD-Data outage closed form-1} in Lemma \ref{Lem: FDD data outage probability}.

We now consider the second aforementioned case, i.e., $b_7\geq b_8+b_9\lVert\mathbf{\hat{h}}_{UT}\rVert^2$. Following similar steps as in the first case, we obtain \eqref{Eq: FDD-Data outage closed form-2} and \eqref{Eq: FDD-Data outage closed form-3} in Lemma \ref{Lem: FDD data outage probability} (the detailed steps have been omitted for matters of space economy). At this stage, the outage probability in \eqref{Eq: FDD-Data outage-1} is obtained as the sum of \eqref{Eq: FDD-Data outage closed form-1}, \eqref{Eq: FDD-Data outage closed form-2}, and \eqref{Eq: FDD-Data outage closed form-3}, and this concludes the proof.

\bibliographystyle{IEEEtran}
\bibliography{reference}

\end{document}